\documentclass{llncs}

\usepackage{microtype,comment}

\title{Contracting to a Longest Path in $H$-Free Graphs}

\author{Walter Kern\inst{1} \and Dani\"el Paulusma\inst{2}\thanks{Supported by Research Project Grant RPG-2016-258 from the Leverhulme Trust.}}

\institute{Department of Applied Mathematics, University of Twente, The Netherlands
\email{w.kern@twente.nl}
\and
Department of Computer Science, Durham University, Durham, UK
\email{daniel.paulusma@durham.ac.uk}}

\newtheorem{oproblem}{Open Problem}
\usepackage{enumerate}
\usepackage{tikz}
\usepackage{boxedminipage,amsmath}
\usetikzlibrary{arrows,shapes,calc}

\DeclareMathOperator{\dist}{dist}

\newcommand{\problemdef}[3]{
	\begin{center}
		\begin{boxedminipage}{.99\textwidth}
			\textsc{{#1}}\\[1pt]  
			\begin{tabular}{ r p{0.8\textwidth}}
				\textit{~~~~Instance:} & {#2}\\
				\textit{Question:} & {#3}
			\end{tabular}
		\end{boxedminipage}
	\end{center}
}

\newcommand{\NP}{{\sf NP}}

\newcommand\displaycase[1]{\textcolor{darkgray}{\sffamily\bfseries\mathversion{bold}#1}}

\newcounter{ctrclaim}[theorem]
\newcounter{ctrcase}[theorem]
\newcommand{\clm}[1]{\setcounter{ctrcase}{0}\medskip\phantomsection\refstepcounter{ctrclaim}\noindent\displaycase{Claim \thectrclaim. }{\em #1}\\}

\newcommand \dia{\hfill{$\diamond$}}

\begin{document}
\maketitle
\setcounter{footnote}{0}

\begin{abstract}
We prove two dichotomy results for detecting long paths as patterns in a given graph.
The \NP-hard problem {\sc Longest Induced Path} is to determine the longest induced path in a graph.
The \NP-hard problem {\sc Longest Path Contractibility} is to determine the longest path to which a graph can be contracted to.
By combining known results with new results we completely classify the computational complexity of both problems for $H$-free graphs. Our main focus is on the second problem, for which we design a general contractibility technique that enables us to reduce the problem to a matching problem.
\end{abstract}

\section{Introduction}\label{s-intro}
The {\sc Hamiltonian Path} problem, which is to decide if a graph has a hamiltonian path, is one of the best-known problems in Computer Science and Mathematics. A more general variant of this problem is that of determining the length of a longest path in a graph. Its decision version {\sc Longest Path} is equivalent to deciding if a graph can be modified into the $k$-vertex path $P_k$ for some given integer~$k$ by using vertex and edge deletions. Note that an alternative formulation of {\sc Hamilton Path} is that of deciding if a graph can be modified into a path (which must be $P_n$) by using only edge deletions.
As such, these problems belong to a wide range of graph modification problems where we seek to modify a given graph~$G$ into some graph~$F$ from some specified family of graphs ${\cal F}$ by using some prescribed set of graph operations.
As {\sc Hamiltonian Path} is \NP-complete (see~\cite{GJ79}), {\sc Longest Path} is \NP-complete as well. The same holds for the problem {\sc Longest Induced Path}~\cite{GJ79}, which is to decide if a graph~$G$ contains an induced path of length at least~$k$, that is, if $G$ can be modified into a path~$P_k$ for some given integer~$k$ by using only vertex deletions.

Here we mainly focus on the variant of the above two problems corresponding to another central graph operation, namely  
edge contraction. This variant plays a role in many graph-theoretic problems, in particular {\sc Hamilton Path}~\cite{HV78,HV81}. 
The \emph{contraction} of an edge~$uv$ of a graph~$G$ deletes the vertices $u$ and $v$ and replaces them by a new vertex made adjacent to precisely those vertices that were adjacent to~$u$ or~$v$ in~$G$ (without introducing self-loops or multiple edges).
A graph~$G$ contains a graph~$G'$ as a {\it contraction} if~$G$ can be modified into~$G'$ 
by a sequence of edge contractions. 

\problemdef{Longest Path Contractibility}{a connected graph $G$ and a positive integer~$k$.}{does $G$ contain $P_k$ as a contraction?}
The {\sc Longest Path Contractibility} problem is \NP-complete as well~\cite{BV87}. Due to the computational hardness of {\sc Longest Path}, {\sc Longest Induced Path} and {\sc Longest Path Contractibility} it is natural to restrict the input to special graph classes. We briefly discuss some known complexity results for the three problems under input restrictions.

A common property of most of the studied graph classes is that they are  {\it hereditary}, that is, they are closed under vertex deletion. As such, they can be characterized by a family of forbidden induced subgraphs. In particular, a graph is {\it $H$-free} if it does not contain a graph~$H$ as an induced subgraph, and a graph class is {\it monogenic} if it consists of all $H$-free graphs for some graph $H$.
Hereditary graph classes defined by a small number of forbidden induced subgraphs, such as monogenic graph classes, are well studied, as evidenced by studies on (algorithmic and structural) decomposition theorems (e.g. for bull-free graphs~\cite{Ch12} or claw-free graphs~\cite{CS05,HMLW11}) and surveys for specific graph problems (e.g. for {\sc Colouring}~\cite{GJPS17,RS04}).

All the known \NP-hardness results for {\sc Hamiltonian Path} carry over to {\sc Longest Path}. For instance, it is known that
{\sc Hamiltonian Path} is \NP-complete for chordal bipartite graphs and strongly chordal split graphs~\cite{Mu96}, line graphs~\cite{Be81} and planar graphs~\cite{GJT76}. Unlike for {\sc Hamiltonian Path}, there are only a few hereditary graph classes for which the {\sc Longest Path} problem is known to be polynomial-time solvable; see, for example~\cite{UU07}. 
In particular, 
{\sc Longest Path} is polynomial-time solvable for circular-arc graphs~\cite{MB14}, distance-hereditary graphs~\cite{GHK13},
 and cocomparability graphs~\cite{IN13,MC12}.
The latter result generalized the corresponding results for bipartite permutation graphs~\cite{UV07} and interval graphs~\cite{IMN11}.
The few graph classes for which the {\sc Longest Induced Path} problem is known to be polynomial-time solvable include the classes of $k$-chordal graphs~\cite{Ga02,IOY08}, AT-free graphs~\cite{KMT03}, graphs of bounded clique-width~\cite{CMR00} (see also~\cite{KMT03}) and graphs of bounded mim-width~\cite{JKT17}. Finding a longest induced path in an $n$-dimensional hypercube is known as the {\sc Snake-in-the-Box} problem~\cite{Ka58}, which has been well studied.\footnote{The complexity status of {\sc Snake-in-the Box}  is still open. A table of world records for small values of $n$ can be found at http://ai1.ai.uga.edu/sib/sibwiki/doku.php/records.}

Unlike the {\sc Longest Path} and {\sc Longest Induced Path} problems, {\sc Longest Path Contractibility} is \NP-complete even for {\it fixed}~$k$ (that is,~$k$ is not part of the input). In order to explain this,  
let $F$-{\sc Contractibility} be the problem of deciding if a graph~$G$ contains some fixed graph~$F$ as a contraction. 
The complexity classification of $F$-{\sc Contractibility} is still open (see~\cite{BV87,LPW08,LPW08b,HKPST12}), but
Brouwer and Veldman~\cite{BV87} showed that already $P_4$-{\sc Contractibility} and $C_4$-{\sc Contractibility} are  \NP-complete (where $C_k$ denotes the $k$-vertex cycle). In fact, $P_4$-{\sc Contractibility} problem is \NP-complete even for $P_6$-free graphs~\cite{HPW09}, whereas
Heggernes et al.~\cite{HHLP14} showed that $P_6$-{\sc Contractibility} is \NP-complete for bipartite graphs.\footnote{In~\cite{HHLP14}, the problem is equivalently formulated as a graph modification problem: can a graph $G$ be modified into a graph~$F$ from some specified family~${\cal F}$ by using at most $\ell$ edge contractions for some given integer~$\ell\geq 0$? We refer to for instance~\cite{ALSZ17,AST17,AH83,BGHP14,CG13,Ep09,GHP13,GM13,HHLLP14,HHLP13,LMS13,MOR13,WAN81,WAN83}, for both classical and fixed-parameter tractibility results for various 
families~${\cal F}$ including the family of paths, which form the focus in this paper.}
The latter result was improved to $k=5$ in~\cite{DP17}. Moreover, $P_7$-{\sc Contractibility} is \NP-complete for line graphs~\cite{FKP13}. Hence, {\sc Longest Path Contractibility} is \NP-complete for all these graph classes as well. 
On the positive side, {\sc Longest Path Contractibility} is polynomial-time solvable for $P_5$-free graphs~\cite{HPW09}.

Our interest in the {\sc Longest Induced Path} problem also stems from a close relationship to a vertex partition problem, which played  a central role in the graph minor project of Robertson and Seymour~\cite{RS95}, as we will explain.

\subsection*{Our Results}

We first give a dichotomy for {\sc Longest Induced Path} using known results for {\sc Hamiltonian Path} and some straightforward observations (see Section~\ref{s-pre} for a proof). Our main result is a dichotomy for {\sc Longest Path Contractibility}. We use `+' to denote the disjoint union of two graphs, and a {\em linear forest} is the disjoint union of one or more paths.

\begin{theorem}\label{t-main0}
Let $H$ be a graph. If $H$ is a linear forest, then {\sc Longest Induced Path} restricted to $H$-free graphs is polynomial-time solvable; otherwise it is \NP-complete.
\end{theorem}

\begin{theorem}\label{t-main}
Let $H$ be a graph. If $H$ is an induced subgraph of  $P_2+P_4$, $P_1+P_2+P_3$, $P_1+P_5$ or $sP_1+P_4$ for some $s\geq 0$, then {\sc Longest Path Contractibility} restricted to $H$-free graphs is polynomial-time solvable; otherwise it is \NP-complete.
\end{theorem}
Theorem~\ref{t-main} shows that {\sc Longest Path Contractibility} is polynomial-time solvable for $H$-free graphs only for some specific linear forests~$H$. This is in contrast to the situation for {\sc Longest Induced Path}, as shown by Theorem~\ref{t-main0}.
 To extend the aforementioned results from~\cite{DP17,FKP13,HHLP14,HPW09} for {\sc Longest Path Contractibility} to the full classification given in Theorem~\ref{t-main} we do as follows.

First, in Section~\ref{s-poly}, we prove the four new polynomial-time solvable cases of Theorem~\ref{t-main}. In each of these cases
$H$ is a linear forest, and proving these cases requires the most of our analysis.\footnote{This is in line with research for other graph problems restricted to $H$-free graphs. In fact, classes of $H$-free graphs, where $H$ is a linear forest are still poorly understood. There is a whole range of graph problems, e.g. {\sc Independent Set}, {\sc $3$-Colouring}, {\sc Feedback Vertex Set}, {\sc Odd Cycle Transversal}, and {\sc Dominating Induced Matching},  for which it is not known if they are \NP-complete on
$P_k$-free graphs for some integer~$k$, such that they are \NP-complete on 
$P_k$-free graphs (see~\cite{BDFJP}).}
Every linear forest~$H$ is $P_r$-free for some suitable value of $r$ and $P_r$-free graphs do not contain $P_r$ as a contraction. Hence, it suffices to prove that for each $1\leq k\leq r-1$, the $P_k$-{\sc Contractibility} problem is polynomial-time solvable for $H$-free graphs for each of the four linear forests listed in Theorem~\ref{t-main}. 
In fact, as $P_3$-{\sc Contractibility} is trivial, we only have to consider the cases where $4\leq k\leq r-1$.
Our general technique for doing this is based on transforming an instance of $P_k$-{\sc Contractibility} for $k\geq 5$ into a polynomial number of instances of $P_{k-1}$-{\sc Contractibility} until $k=4$. 

For $k=4$ we cannot apply this transformation, as this case - as we outline below - is closely related to 
the $2$-{\sc Disjoint Connected Subgraphs} problem. This problem takes as input a triple $(G,Z_1,Z_2)$, where $G$ is a graph with two disjoint subsets~$Z_1$ and~$Z_2$ of~$V(G)$.
It asks if~$V(G)\setminus (Z_1\cup Z_2)$ has a partition into sets~$S_1$ and~$S_2$,  such that $Z_1\cup S_1$ and $Z_2\cup S_2$ induce connected subgraphs of~$G$. Robertson and Seymour~\cite{RS95} proved that the more general problem $k$-{\sc Disjoint Connected Subgraphs}  (for $k$ subsets $Z_i$) is polynomial-time solvable as long as the union of the sets $Z_i$ has constant 
size.\footnote{If every $Z_i$ has size~2, then we obtain the well-known {\sc $k$-Disjoint Paths} problem.}  
However, in our context, $Z_1$ and $Z_2$ may have arbitrarily large size. In that case, $2$-{\sc Disjoint Connected Subgraphs} is \NP-complete even if $|Z_1|=2$ (and only $Z_2$ is large)~\cite{HPW09}.

To work around this obstacle, we use the fact~\cite{HPW09} that the two outer vertices of the~$P_4$, to which the input graph~$G$ must be contracted, may correspond to single vertices $u$~and~$v$~of~$G$. We  then ``guess''  $u$ and $v$ to obtain an instance $(G-\{u,v\},N(u),N(v))$ of $2$-{\sc Disjoint Subgraphs}. That is, we seek for a partition of $(V(G)\setminus \{u,v\})\setminus ((N_u)\cup N(v))$ into sets $S_u$ and $S_v$, such that $N(u)\cup S_u$ and $N(v)\cup S_v$ are connected. The latter implies that we can contract these two sets to single vertices corresponding to the two middle vertices of the~$P_4$.

After guessing $u$ and $v$ we exploit their presence, together with the $H$-freeness of~$G$, for an extensive analysis of the structure of $S_u$ and $S_v$ of a potential solution $(S_u,S_v)$. To this end we introduce 
in Section~\ref{s-p4} some general terminology and first show how to check in general for solutions in which
the part of $S_u$ or $S_v$  that ensures connectivity of $N(u)\cup S_u$ or $N(v)\cup S_v$, respectively, has bounded size. 
We call such solutions constant. If we do not find 
a constant solution, then we exploit their absence. For the more involved cases we show that in this way we can branch to a polynomial number of instances of a standard matching problem.

In Section~\ref{s-hard} we prove the new \NP-completeness results. In particular, we prove that
$P_k$-{\sc Contractibility}, for some suitable value of $k$, is \NP-complete for bipartite graphs of large girth, strengthening the known result for bipartite graphs of~\cite{HHLP14}. 

In Section~\ref{s-classification} we show how to combine our new polynomial-time and \NP-hardness results with the known \NP-completeness results for $K_{1,3}$-free graphs~\cite{FKP13} and  $P_6$-free graphs~\cite{HPW09} in order to obtain Theorem~\ref{t-main}.

In Section~\ref{s-cycle}, we briefly discuss the cycle variant of our problem, called the
{\sc Longest Cycle Contractibility} problem~\cite{Bl82,Ha99,Ha02}.
Its complexity classification for $H$-free graphs is still incomplete, but we show that it differs from the classification of {\sc Longest Path Contractibility} for $H$-free graphs.

In Section~\ref{s-con} we pose some open problems. In particular, the complexity classification of {\sc Longest Path} is still open for $H$-free graphs, and we describe the state-of-art for this problem.

\section{Preliminaries}\label{s-pre}

In Section~\ref{s-gt} we give some general graph-theoretic terminology and a helpful lemma for $P_4$-free graphs.
In Section~\ref{s-t1} we give a short proof of Theorem~\ref{t-main0}.
In Section~\ref{l-et} we give some terminology related to edge contractions.

\subsection{General Terminology and a Lemma for $P_4$-Free Graphs}\label{s-gt}

We consider finite undirected graphs with no self-loops.
Let $G=(V,E$ be a graph. Let $S\subseteq V$. Then $G[S]=(S,\{uv\in E\; |\; u,v\in S\})$ denotes the subgraph of $G$ {\it induced} by $S$. We say that $S$ is {\it connected} if $G[S]$ is connected. We may write $G-S=G[V\setminus S]$.
The \emph{neighbourhood} of $v\in V$ is the set $N(v)=\{u\; |\; uv\in E\}$ and the {\it closed neighbourhood} is
$N[v]=N(v)\cup \{v\}$.
The {\it length} of a  path~$P$ is its number of edges.
The {\it distance} $\dist_G(u,v)$ between vertices $u$ and $v$ is the length of a shortest path between them.
Two disjoint sets $S, T\subset V$ are {\it adjacent} if there is at least one edge between them; $S$ and~$T$ are
{\it (anti)complete} to each other if every vertex of $S$ is (non)adjacent to every vertex of~$T$. 
The set~$S$ {\it covers} $T$ if every vertex of $T$ has a neighbour in $S$. 
The {\it subdivision} of an edge~$e=uv$ in $G$ replaces $e$ by a new vertex $w$ and two new edges $uw$ and $wv$.

A graph $G$ is {\it $H$-free} for some other graph~$H$ if $G$ does not contain $H$ as an induced subgraph.
For a set $H_1,\ldots,H_p$ of graphs, $G$ is {\it $(H_1,\ldots,H_p)$-free} if $G$ is $H_i$-free for $i=1,\ldots,p$.
A graph is {\it complete bipartite} if it consists of a single vertex or its vertex set can be partitioned into two independent sets $A$ and $B$  that are complete to each other. The {\it claw} $K_{1,3}$ is the complete bipartite graph with $|A|=1$ and 
$|B|=3$.
The graph~$K_n$ is the complete graph on $n$ vertices.

The {\em disjoint union} $G_1+\nobreak G_2$ of two vertex-disjoint graphs~$G_1$ and~$G_2$ is the graph $(V(G_1)\cup V(G_2), E(G_1)\cup E(G_2))$; the disjoint union of~$r$ copies of a graph~$G$ is denoted~$rG$.
A {\it forest} is a graph with no cycles. 
A {\it linear forest} is a forest of maximum degree at most~2, that is, a disjoint union of one or more paths. 
The {\it join} operation~$\times$  adds an edge between every vertex of $G_1$ and every vertex of $G_2$.  
A graph $G$ is a {\it cograph} if $G$ can be generated from $K_1$ by a sequence of join and disjoint union operations.
A graph is a cograph if and only if it is $P_4$-free (see, e.g.,~\cite{BLS99}).
The following well-known lemma follows from this fact and the definition of a cograph. In particular, to prove that a connected $P_4$-free graph~$G$ has a spanning complete bipartite graph with partition classes $A$ and $B$, we can do as follows:
take the complement $\overline{G}=(V,\{uv\; |\; uv\not \in E\; \mbox{and}\; u\neq v\}$ of $G$ and put the vertex set of one connected component of~$\overline{G}$ in~$A$ and all the other vertices of $\overline{G}$ in~$B$.

\begin{lemma}\label{l-p4}
Every connected $P_4$-free graph on at least two vertices has a spanning complete bipartite subgraph, 
which can be found in polynomial time.
\end{lemma}

We remind the reader of the following notions.The {\it girth} of a graph $G$ that is not a forest is the number of vertices in a shortest induced cycle of $G$.
The {\it line graph} $L(G)$ of a graph $G=(V,E)$ has $E$ as vertex set and there is an edge between two vertices $e_1$ and $e_2$ of $L(G)$ if and only if $e_1$ and $e_2$ have a common end-vertex in $G$. Every line graph is readily seen to be $K_{1,3}$-free.

\subsection{The Proof of Theorem~\ref{t-main0}}\label{s-t1}

We now present a short proof for Theorem~\ref{t-main0}. We start with the following lemma.

\begin{lemma}\label{l-girthpath}
Let $p\geq 3$ be some constant. Then {\sc Longest Induced Path} is \NP-complete for graphs of girth at least~$p$.
\end{lemma}

\begin{proof}
We reduce from {\sc Hamiltonian Path}. Let $G$ be a graph on $n$ vertices. We subdivide each 
edge~$e$ of $G$ exactly once and denote the set of new vertices $v_e$ by $V'$.
We denote the resulting graph by $G'$ and note that $G'$ is bipartite with partition classes $V$ and $V'$.
We claim that $G$ has a Hamiltonian path if and only if $G'$ has an induced path of length $2n-2$.

First suppose that $G$ has a Hamiltonian path $u_1u_2\cdots u_n$. Then the path on vertices $u_1, v_{u_1u_2}, u_2,
\ldots, v_{u_{n-1}u_n}, u_n$ is an induced path of length~$2n-2$ in $G'$.
Now suppose that $G'$ has an induced path $P'$ of length $2n-2$. Then either $P'$ starts and finished with a vertex of $V$, or $P'$
starts and finishes with a vertex of $V'$. In the first case $P'$ contains $n$ vertices of $G$, so $P$ contains all vertices $u_1,\ldots,u_n$ of $G$, say in this order. Then $u_1u_2\cdots u_n$ is a Hamiltonian path of $G$.
In the second case $P'$ contains $n-1$ vertices of $V$, say vertices $u_1,\ldots,u_{n-1}$ in that order. As $P'$ is an induced path and vertices of $V'$ are only adjacent to vertices of $V$, this means that the end-vertices of $P'$ are both adjacent to $u_n$.
Hence, we find that $u_1u_2\cdots u_n$ is a Hamiltonian path of $G$  (and the same holds for $u_nu_1\cdots u_{n-1}$).

We note that the girth of $G'$ is twice the girth of $G$. Hence, we obtain the result by applying this trick sufficiently many times. \qed
\end{proof}

We also need the following lemma.

\begin{lemma}\label{l-line}
The {\sc Longest Induced Path} problem is \NP-complete for line graphs.
\end{lemma}

\begin{proof}
We reduce from {\sc Hamiltonian Path}. Let $G=(V,E)$ be a graph on $n$ vertices. We construct the line graph $L(G)$ of $G$. We claim that $G$ has a Hamiltonian path if and only if $L(G)$ has an induced path on $n-1$ vertices.
First suppose that $P$ is a Hamiltonian path in $G$. Then the edges of $P$ form an induced path of length
$n-1$ in  $L(G)$. Now suppose that $L(G)$ has an induced path~$\tilde{P}$ 
on $n-1$ vertices.
 Let $e_1, \dots, e_{n-1}$ be the $n-1$ edges of $\tilde{P}$ in that order. 
 As $\tilde{P}$ is induced in 
 $L(G)$, no two edges $e_i$ and $e_j$
with $i<j$ have a vertex $v \in V$ in common unless $j=i+1$. 
Hence, $P=\{e_1, \dots, e_{n-1}\}$ must be a Hamiltonian path in $G$. \qed
\end{proof}

We are now ready to prove Theorem~\ref{t-main0}.

\medskip
\noindent
{\bf Theorem~\ref{t-main0}. (restated)}
{\it Let $H$ be a graph. If $H$ is a linear forest, then
{\sc Longest Induced Path} restricted to $H$-free graphs is polynomial-time solvable; otherwise it is \NP-complete.}

\begin{proof}
Let $G$ be an $H$-free graph.
First suppose that $H$ is a linear forest. Then there exists a constant~$k$ such that $H$ is an induced subgraph of $P_k$.
This means that the length of a longest induced path of $G$ is at most $k-1$. Hence, we can determine a longest path in $G$ in 
$O(n^{k-1})$ time by brute force.

Now suppose that $H$ is not a linear forest. First assume that $H$ contains a cycle.  
Let $g$ be the girth of $H$. We set $p=g+1$. Then the class of $H$-free graphs contains the class
of graphs of girth at least~$p$. Hence, we can use Lemma~\ref{l-girthpath} to find that {\sc Longest Induced Path} is \NP-complete for $H$-free graphs. 
Now assume that $H$ contains no cycle. As $H$ is not a linear forest, $H$ must be a forest with at least one vertex of degree at least~3. Then the class of $H$-free graphs contains the class of $K_{1,3}$-free graphs. Recall that every line graph is $K_{1,3}$-free.
Hence, the class of line graphs is contained in the class of $H$-free graphs. Then we can use Lemma~\ref{l-line} to find that {\sc Longest Induced Path} is \NP-complete for $H$-free graphs. \qed
\end{proof}

\subsection{Terminology Related to Edge Contractions}\label{l-et}
Recall that the contraction of an edge~$uv$ of a graph~$G$ is the operation that deletes $u$ and $v$ from $G$ and replaces them by a new vertex made adjacent to precisely those vertices that were adjacent to~$u$ or~$v$ in~$G$ (without introducing self-loops or multiple edges).
We denote the graph obtained from a graph $G$ by contracting $e=uv$ by $G/e$. We may denote the resulting vertex by $u$ (or $v$) again and say that we {\it contracted $e$} on~$u$ (or~$e$ on~$v$). 

Recall also that a graph $G$ contains a graph $H$ as a contraction if $G$ can be modified into $H$ via a sequence of edge contractions. Alternatively, a graph~$G$ contains a graph~$H$ as a contraction if and only if for every vertex~$x\in V(H)$ there exists a nonempty subset $W(x)\subseteq V(G)$ of vertices in~$G$ such that:
\begin{itemize}
\item [(i)] $W(x)$ is connected;  
\item [(ii)] the set ${\cal W}=\{W(x)\; |\; x\in V_H\}$ is a partition of~$V(G)$; and
\item [(iii)] for every $x_i,x_j\in V(H)$, $W(x_i)$ and~$W(x_j)$ are adjacent in $G$ if and only if~$x_i$ and~$x_j$ are adjacent in~$H$.
\end{itemize}
By contracting the vertices in each~$W(x)$ to a single vertex we obtain the graph~$H$.
The set~$W(x)$ is called an 
$H$-{\it witness bag} of~$G$ for~$x$. The set~${\cal W}$ is called an {\it $H$-witness structure} of~$G$ (which does not have to be unique).
A pair of 
(non-adjacent) vertices $(u,v)$ of a graph~$G$ is $P_k$-\emph{suitable} for some integer $k\geq 3$ if and only if~$G$ has a $P_k$-witness structure~${\cal W}$ with $W(p_1)=\{u\}$ and 
$W(p_k)=\{v\}$, where $P_k=p_1\dots p_k$; see Figure~\ref{f-p4witness} for an example.

\begin{figure}
  \centering
  \includegraphics[scale=1]{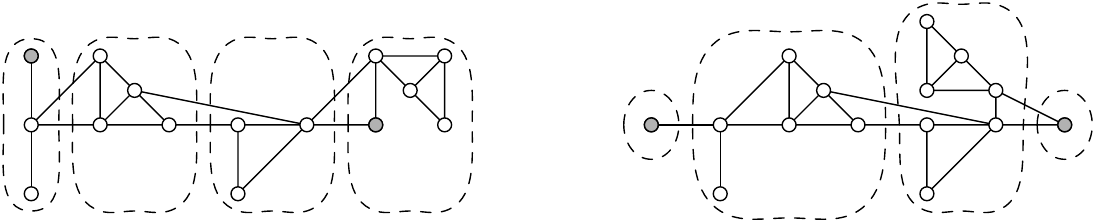}
  \caption{Two $P_4$-witness structures of a graph; the grey vertices form a $P_4$-suitable pair~\cite{HPW09}.}\label{f-p4witness}
\end{figure}

The following known lemma shows why $P_k$-suitable pairs are of importance.

\begin{lemma}[\cite{HPW09}]\label{l-outer}
For $k\geq 3$, a graph~$G$ contains $P_k$ as a contraction if and only if~$G$ has a $P_k$-suitable pair.
\end{lemma}
Lemma~\ref{l-outer} leads to the following auxiliary problem, 
where $k\geq 3$ is a fixed integer, that is, $k$ is not part of the input.

\problemdef{$P_k$-Suitability}{a connected graph $G$ and two non-adjacent vertices $u,v$.}{is $(u,v)$ a $P_k$-suitable pair?}
The next, known observation follows from the fact
 that $P_k$-{\sc Contractibility} is trivial for 
 $k\leq 2$, 
 whereas for $k=3$ we can use Lemma~\ref{l-outer} combined with the observation that $P_3$-{\sc Suitability} is polynomial-time solvable (two non-adjacent vertices $u$, $v$ form a $P_3$-suitable pair in a connected graph $G$ if and only if $G-\{u,v\}$ is connected).

\begin{lemma}\label{l-trivial}
For $k\leq 3$, {\sc $P_k$-Contractibility}  can be solved in polynomial time.
\end{lemma}

We now show the following lemma, which will be helpful for proving our results.

\begin{lemma}\label{l-reduce}
Let $k\geq 4$ and let $(G,u,v)$ be an instance of $P_k$-{\sc Suitability} with $u$ and $v$ at distance $d > k$. Let $P$ be a shortest path from $u$ to $v$. Then $(G,u,v)$ can be reduced in polynomial time to $d-2$ instances $(G/e,u,v)$, one for each edge~$e\in E(P)$ that is not incident to $u$ and $v$, with $\dist(u,v) =d-1$, such that $(G,u,v)$ is a yes-instance if and only if at least one of the new instances $(G/e,u,v)$ is a yes-instance of $P_k$-{\sc Suitability}.
\end{lemma}

\begin{proof}
First suppose that $(G,u,v)$ is a yes-instance of $P_k$-{\sc Suitability}.
Then $G$ has a $P_k$-witness structure ${\cal W}$ with $W(p_1)=\{u\}$ and $W(p_k)=\{v\}$. As $d\geq k$, at least one bag of ${\cal W}$ will contain both end-vertices of an edge~$e$ of $P$.   
Then contracting $e$ yields a $P_{k}$-witness structure~${\cal W}'$ for $(G/e,u,v)$ with $W'(p_1)=\{u\}$ and $W'(p_k)=\{v\}$. 
As $W(p_1)$ and $W(p_k)$ only contain $u$ and $v$, respectively, the end-vertices of $e$ belong to some bag $W(p_i)$ with 
$2\leq i \leq k-1$. Hence, $e$ is not incident to $u$ and $v$.

Now suppose that $P$ contains an edge $e$ not incident to $u$ and $v$ such that $(G/e,u,v)$ is a yes-instance of $P_k$-{\sc Suitability}. Then $G/e$ has a $P_k$-witness structure ${\cal W}'$ with $W'(p_1)=\{u\}$ and $W'(p_k)=\{v\}$. Let $e=st$ and say that we contracted $e$ on $s$. As $e$ is not incident to $u$ and $v$, we find that $\{s,t\}\cap \{u,v\}=\emptyset$.
Hence, $s$ belongs to some bag $W'(p_i)$ with $2\leq i \leq k-1$. Then in $W'(p_i)$ we uncontract $e$ (so the new bag will contain both $s$ and $t$). This yields a $P_k$-witness structure ${\cal W}$ of $G$ with $W(p_1)=\{u\}$ and $W(p_k)=\{v\}$. \qed
\end{proof}

\noindent
In our polynomial-time algorithms for constructing $P_k$-witness structures (to prove Theorem~\ref{t-main}) we put vertices in certain sets that we  then try to extend to $P_k$-witness bags (possibly via branching) and we will often apply the following rule:

\medskip
\noindent
{\tt Contraction Rule.} If two adjacent vertices 
$s$ and $t$ end up in the same bag of some potential $P_k$-witness structure, then contract the edge~$st$.

\medskip
\noindent
For a graph $G=(V,E)$, 
we say that 
we {\it apply the {\tt Contraction Rule} on some set} $U\subseteq V$ if we contract every edge
in $G[U]$. 
The advantage of applying this rule is that we obtain a smaller instance and that we can exploit the fact that the resulting set $G[U]$ has
become independent.

It is easy to construct examples that show that a class of $H$-free graphs is not closed under contraction if $H$ contains a vertex of degree at least~3 or a cycle. 
However, all polynomial-time solvable cases of Theorem~\ref{t-main} involve forbidding a linear forest~$H$. 
The following known lemma, which is readily seen, shows that the {\tt Contraction Rule} does preserve $H$-freeness as long as $H$ is a linear forest. Hence, we can safely apply the rule in our proofs of the polynomial-time solvable cases of Theorem~\ref{t-main}.

\begin{lemma}\label{l-contract}
Let $H$ be a linear forest and let $G$ be an $H$-free graph. Then the graph obtained from $G$ after contracting an edge is also $H$-free.
\end{lemma}

\section{The Polynomial-Time Solvable Cases of Theorem~\ref{t-main}}\label{s-poly}

In this section we prove that {\sc Longest Path Contractibility} is polynomial-time solvable for 
$H$-free graphs if $H=P_2+P_4$ (Section~\ref{s-p2p4}), $H=P_1+P_2+P_3$ (Section~\ref{s-p1p2p3}), $H=P_1+P_5$ (Section~\ref{s-p1p5}) and $H=sP_1+P_4$ for every integer~$s\geq 0$ (Section~\ref{s-sp1p4}).
To solve {\sc Longest Path Contractibility} in each of these cases  we will eventually check if the input graph can be contracted to $P_4$. This turns out to be the hardest situation to deal with in our proofs.
 Due to Lemma~\ref{l-outer}, we can solve it by checking 
for each pair of distinct vertices $u$, $v$ with $N(u)\cap N(v)=\emptyset$
 if $(G,u,v)$ is a yes-instance of $P_4$-{\sc Suitability} (note that for any other pair $u$, $v$, we have that $(G,u,v)$ is a no-instance 
 of $P_4$-{\sc Suitability}).
 In Section~\ref{s-p4} we first provide a general framework by introducing some additional terminology and one general result for
solving $P_4$-{\sc Suitability}.

\subsection{On Contracting a Graph to $\mathbf{P_4}$}\label{s-p4}

Let $(G,u,v)$ be an instance of $P_4$-{\sc Suitability}.  
For every $P_4$-witness structure of~$G$ with $W(p_1)=\{u\}$ and $W(p_4)=\{v\}$ (if it exists),
every neighbour of $u$ belongs to $W(p_2)$ and every neighbour of $v$ belongs to $W(p_3)$.
Throughout our proofs we let $T(u,v)=V(G)\setminus (N[u]\cup N[v])$ denote the set of remaining vertices of $G$, which still need to be placed in either $W(p_2)$ or $W(p_3)$. We write $T=T(u,v)$ if no confusion is possible.
We say that a partition $(S_u,S_v)$ of $T$ is a {\it solution} for $(G,u,v)$ if
$N(u)\cup S_u$ and $N(v)\cup S_v$ are both connected. Hence, a solution $(S_u,S_v)$ for $(G,u,v)$ corresponds to a 
$P_4$-witness structure ${\cal W}$ of $G$, where $W(p_1)=\{u\}$, $W(p_2)=N(u)\cup S_u$, $W(p_3)=N(v)\cup S_v$ and
$W(p_4)=\{p_4\}$. 
A solution $(S_u,S_v)$ for $(G,u,v)$ is {\it $\alpha$-constant} for some constant~$\alpha\geq 0$ if the following 
holds: either $S_u$ contains a set $S_u'$ of size $|S_u'|\leq\alpha$ such that $N(u)\cup S_u'$ is connected, or $S_v$ contains a 
set $S_v'$ of size $|S_v'|\leq\alpha$ such that $N(v)\cup S_v'$ is connected. 
We prove the following lemma.

\begin{lemma}\label{l-constant}
Let $(G,u,v)$ be an instance of $P_4$-{\sc Suitability}. For every constant $\alpha\geq 0$, it is possible to check in $O(n^{\alpha+2})$ time if $(G,u,v)$ has an $\alpha$-constant solution.
\end{lemma}

\begin{proof}
We first do the following check for vertex~$u$. For each set~$S$ of size $|S|\leq \alpha$ we check if $N(u)\cup S$ is connected and if  every vertex of $N(v)$ is in the same connected component $D$ of the subgraph of $G$ induced by $(T\setminus S)\cup N(v)$. If so, then we put all vertices of $T \setminus V(D)$ in $S_u$ and all vertices of $T\cap V(D)$ in $S_v$. As $G$ is connected, this yields a solution $(S_u,S_v)$ for $(G,u,v)$. This takes $O(n^2)$ time for each set~$S$. As the number of sets~$S$ is $O(n^\alpha)$, the total running time is $O(n^{\alpha+2})$. We can do the same check in $O(n^{\alpha+2})$ time for vertex~$v$. This proves the lemma. \qed
\end{proof}

Let  $(S_u,S_v)$ be a solution for an instance $(G,u,v)$ of $P_4$-{\sc Suitability} that is not $7$-constant (the value $\alpha=7$ comes from our proofs).
If $G[S_u]$ and $G[S_v]$ each contain at least one edge, then $(S_u,S_v)$ is {\it double-sided}.
 If exactly one of $G[S_u]$, $G[S_v]$ contains an edge, then $(S_u,S_v)$ is {\it single-sided}. If both $S_u$ and $S_v$ are independent sets, then $(S_u,S_v)$ is {\it independent}.

\subsection{The Case $\mathbf{H=P_2+P_4}$}\label{s-p2p4}

We now show that {\sc Longest Path Contractibility} is polynomial-time solvable for $(P_2+P_4)$-free graphs. As mentioned, we will do so via the auxiliary problem $P_k$-{\sc Suitability}.
We first give, in Lemma~\ref{l-top4}, a polynomial-time algorithm for $P_4$-{\sc Suitability} for $(P_2+P_4)$-free graphs. This is the most involved part of our algorithm. As such, we start with an outline of this algorithm.

\medskip
\noindent
{\it Outline of the $P_4$-Suitability Algorithm for $(P_2+P_4)$-free graphs.}\\
We first observe that for an instance $(G,u,v)$, we may assume that~$u$ and~$v$ are of distance at least~3, and consequently, $N(u)\cap N(v)=\emptyset$, and moreover we may assume that $N(u)$ and $N(v)$ are independent.
Recall that $T=V(G)\setminus (N[u]\cup N[v])$. To get a handle on the adjacencies between $T$ and $V(G)\setminus T$ we will apply a (constant) number of branching procedures. For example, we will prove in this way that $G[T]$ may be assumed to be $P_4$-free. Each time we branch we obtain, in polynomial time, a polynomial number of new, smaller
instances of {\sc $P_4$-Suitability} satisfying additional helpful constraints, such that the original instance is a yes-instance if and only if at least one of the new instances is a yes-instance. 
We then consider each new instance separately. 
That is, we either solve, in polynomial time, the problem for each new instance or create a polynomial number of new and even smaller instances via some further branching. 

Our first goal is to check if $(G,u,v)$ has an 
$7$-constant solution. If so then we are done. Otherwise we prove that the absence of $7$-constant solutions implies that $(G,u,v)$ has no double-sided solution either. Hence, it remains to test if 
$(G,u,v)$ has a single-sided solution or an independent solution. We check single-sidedness with respect to $u$ and $v$ independently.
We show that in both cases this leads either to a solution or to a polynomial number of smaller instances, for which we only need to check if they have an independent solution. This will enable us to branch in such a way that afterwards we may assume that $T$ is an independent set and that the solution we are looking for is equivalent to finding a star cover of $N(u)$ and $N(v)$ with centers in $T$. The latter problem reduces to a matching problem, which we can solve in polynomial time.

\begin{lemma}\label{l-top4}
$P_4$-{\sc Suitability} can be solved in polynomial time for $(P_2+P_4)$-free graphs.
\end{lemma}

\begin{proof}
\setcounter{ctrclaim}{0}
Let $(G,u,v)$ be an instance of $P_4$-{\sc Suitability}, where $G$ is a connected $(P_2+P_4)$-free graph. We may assume without loss of generality that $u$ and $v$ are of distance at least 3, that is, $u$ and $v$ are non-adjacent and $N(u)\cap N(v)=\emptyset$; otherwise $(G,u,v)$ is a no-instance. 

Recall that $T=V(G)\setminus (N[u]\cup N[v])$.
Recall also that we are looking for a partition $(S_u,S_v)$ of $T$ that is a solution for $(G,u,v)$, that is, 
$N(u)\cup S_u$ and $N(v)\cup S_v$ must both be connected. 
In order to do so we will construct partial solutions $(S_u',S_v')$, which we try 
to extend to a solution $(S_u,S_v)$ for $(G,u,v)$.
We use the {\tt Contraction Rule} from Section~\ref{s-pre} on $N(u)\cup S_u'$ and $N(v)\cup S_v'$, so that these two sets will become independent.
By Lemma~\ref{l-contract}, the resulting graph will always be $(P_2+P_4)$-free.
 For simplicity, we will denote the resulting instance by $(G,u,v)$ again. 
 After applying the {\tt Contraction Rule} the size of the set $T$ 
will be reduced if a vertex  $t\in T$ was involved in an edge contraction with a vertex from $N(u)$ or $N(v)$.
In that case we say that we {\it contracted $t$ away}.

At the beginning of our algorithm, $S_u'=S_v'=\emptyset$, and we start by applying the {\tt Contraction Rule} on $N(u)$ and $N(v)$. This leads to the following claim.

\clm{\label{c-ind} $N(u)$ and $N(v)$ are independent sets.}\\[-8pt]

\noindent
{\bf Phase 1: Exploiting the structure of $\mathbf{G[T]}$}

\medskip
\noindent
In the first phase of our algorithm, we will look into the structure of $G[T]$.
Suppose $G[T]$ contains an induced $P_4$ on vertices $a_1$, $a_2$, $a_3$, $a_4$. 
If there exists a vertex $t\in N(u)$ not adjacent to any vertex of $\{a_1,a_2,a_3,a_4\}$, then $\{u,t\}\cup \{a_1,a_2,a_3,a_4\}$ induces a $P_2+P_4$ in $G$, a contradiction. Hence, $\{a_1,a_2,a_3,a_4\}$ must cover $N(u)$. Similarly, $\{a_1,a_2,a_3,a_4\}$ must 
cover $N(v)$. Suppose $G[T]$ has another induced $P_4$ on vertices $\{b_1,b_2,b_3,b_4\}$ such that $\{a_1,a_2,a_3,a_4\}\cap \{b_1,b_2,b_3,b_4\}=\emptyset$. By the same arguments, $\{b_1,b_2,b_3,b_4\}$ also covers $N(u)$ and $N(v)$. This means that $N(u)\cup \{a_1,a_2,a_3,a_4\}$ and $N(v)\cup \{b_1,b_2,b_3,b_4\}$ are both connected. 
We put each remaining vertex of $T$ into either $S_u$ or $S_v$ (which is possible, as $G$ is connected).
This yields a ($4$-constant) solution for $(G,u,v)$.

From now on, assume that $G[T]$ contains no induced copy of $P_4$ that is vertex-disjoint from $a_1a_2a_3a_4$ (so, every other induced $P_4$ in $G[T]$ contains at least one vertex of $\{a_1,a_2,a_3,a_4\}$). 
Below we will branch into $O(n^{16})$ smaller instances in which $G[T]$ is $P_4$-free, such 
that $(G,u,v)$ has a solution if and only if at least one of these new instances has a solution.

\medskip
\noindent
{\bf Branching I} ($O(n^{16})$ branches)\\
We branch by considering every possibility for each $a_i$ $(1\leq i\leq 4$) to go into either $S_u$ or $S_v$ for some solution
$(S_u,S_v)$ of $(G,u,v)$ (if it exists). 
We do this vertex by vertex leading to a total of $2^4$ branches.
Suppose we decide to put $a_i$ in $S_u$. 
If $a_i$ is adjacent to a vertex of $N(u)$, then we apply the {\tt Contraction Rule} on $N(u)\cup \{a_i\}$ to contract $a_i$ away.  If $a_i$ is not adjacent to any vertex of $N(u)$, then we do as follows.
For each solution $(S_u,S_v)$ with $a_i\in S_u$, there must exist a shortest path $P_i$ in $G[N(u)\cup S_u]$ from $a_i$ to a 
vertex of $N(u)$ (as $N(u)\cup S_u$ is connected). As $G$ is $(P_2+P_4)$-free, $G$ is $P_7$-free. Hence, $P_i$ 
must have at most six vertices and thus 
at most four inner vertices. We consider all possibilities of choosing at most four vertices of $T$ to belong to $S_u$ as inner vertices of $P_i$. As we may need to do this for $i=1,\ldots,4$, the above leads to a total of $O(n^{16})$ additional branches.

For each branch we do as follows. For $i=1,\ldots,4$ we apply the {\tt Contraction Rule}
on $N(u)\cup \{a_i\} \cup V(P_i)$ to contract $a_i$ and the vertices of $V(P_i)$ away.
We denote the resulting instance by $(G,u,v)$ again. Note that the property (Claim~\ref{c-ind}) that $N(u)$ and $N(v)$ are independent sets is maintained. Moreover, as
every induced $P_4$ in $G[T]$ contained at least one vertex of $\{a_1,\ldots,a_4\}$, the following claim holds now as well.

\clm{\label{c-p4free} $G[T]$ is $P_4$-free.}\\[-8pt]
\noindent
We now prove the following claim.

\clm{\label{c-either} Let $(S_u,S_v)$ be a solution for $(G,u,v)$ that is not $7$-constant. Let 
$t,x_1,x_2$ be three vertices of $T$ with $tx_1\notin E(G)$, $tx_2\notin E(G)$ and $x_1x_2\in E(G)$.
If $t,x_1,x_2$ are in $S_u$, then every neighbour of $t$ in $N(u)$ is adjacent to at least one 
of $x_1,x_2$. If $t,x_1,x_2$ are in $S_v$, then every neighbour of $t$ in $N(v)$ is adjacent to at least one 
of $x_1,x_2$.}\\[-8pt]
{\it Proof of Claim~\ref{c-either}.} We assume without loss of generality that $t,x_1,x_2$ belong to $S_u$. Suppose $t$ has a neighbour $w\in N(u)$ that is not adjacent to $x_1$ and $x_2$. Suppose there exists a vertex $w' \in N(u)$ not  adjacent to any of $t,x_1,x_2$. Then, as $N(u)$ is independent by Claim~\ref{c-ind}, $\{x_1,x_2\}\cup  \{w',u,w,t\}$ is an induced $P_2+P_4$, a contradiction. Hence, $\{t,x_1,x_2\}$ covers $N(u)$. As $N(u)\cup S_u$ is connected, $G[N(u)\cup S_u]$ contains a shortest path $P$ from $t$ to~$x_1$. As $G$ is $(P_2+P_4)$-free, $G$ is $P_7$-free. Hence, $P$ has at most four inner vertices (possibly including $x_2$). As $V(P)\cup \{x_2\} \cup N(u)$ is connected and $|V(P) \cup \{x_2\}|\leq 7$, we find that $(S_u,S_v)$ is a $7$-constant solution, a contradiction. This proves the claim. \dia

\medskip
\noindent
We will use the above claim at several places in our proof, including in the next stage.

\bigskip
\noindent
{\bf Phase 2: Excluding 7-constant solutions and double-sided solutions}

\medskip
\noindent
We first show that we may exclude double-sided solutions if we have no $7$-constant solutions.
 
\clm{\label{c-double}If $(G,u,v)$ has a double-sided solution, then $(G,u,v)$ also has a $7$-constant solution.}\\[-8pt]
{\it Proof of Claim~\ref{c-double}.} For contradiction, assume that $(G,u,v)$ has a double-sided solution $(S_u,S_v)$ but no
$7$-constant solution. By definition, $G[S_u]$ and $G[S_v]$ contain some edges $x_1x_2$ and $x_1'x_2'$, respectively. 
Then $N(u)$ must contain a vertex~$w$ that is not adjacent to $x_1$ and $x_2$; otherwise the two vertices $x_1,x_2$, which belong to $S_u$, cover $N(u)$ and this would imply that $(S_u,S_v)$ is a $2$-constant solution, and thus also a $7$-constant solution.
As $N(u)\cup S_u$ is connected and $N(u)$ is an independent set by Claim~\ref{c-ind}, set~$S_u$ must contain a vertex~$t$ that is adjacent to $w$. Then, by Claim~\ref{c-either}, vertex~$t$ must be adjacent to at least one of $x_1,x_2$, say $x_1$. For the same reason,  $S_v$ contains a vertex $t'$ that is adjacent to at least one of $x_1',x_2'$, say $x_1'$, and to some vertex $w'\in N(v)$ that is not adjacent to $x_1'$ and $x_2'$. 

Let $y\in N(v)$. If no vertex of $\{t,x_1,x_2\}$ is adjacent to $y$, then $\{v,y\}\cup \{u,w,t,x_1\}$ is an induced $P_2+P_4$ in $G$, unless $wy\in E(G)$. However, in that case $\{x_1,x_2\}\cup \{u,w,y,v\}$ is an induced $P_2+P_4$, a contradiction. Hence, $\{t,x_1,x_2\}$ covers $N(v)$. For the same reason we find that $\{t',x_1',x_2'\}$ covers $N(u)$. Then $(G,u,v)$ has a $3$-constant solution $(S_u^*,S_v^*)$ (which is $7$-constant by definition) with $\{t',x_1',x_2'\}\subseteq S_u^*$ and $\{t,x_1,x_2\}\subseteq S_v^*$, a contradiction. This proves the claim. \dia

\medskip
\noindent
Recall that a solution $(S_u,S_v)$ for $(G,u,v)$ is single-sided if exactly one of $G[S_u]$, $G[S_v]$ contains an edge and independent if $S_u,S_v$ are both independent sets.
We now do as follows. First we check in polynomial time if $(G,u,v)$ has a $7$-constant solution by using Lemma~\ref{l-constant}.
If so, then we are done. From now on assume that $(G,u,v)$ has no $7$-constant solution.
Then, by Claim~\ref{c-double} it follows that $(G,u,v)$ has no double-sided solution. 

From the above, it remains to check if $(G,u,v)$ has a single-sided solution or an independent solution. If $(G,u,v)$ has a single-sided solution $(S_u,S_v)$ that is not independent, then either $S_u$ or $S_v$ is independent.
Our algorithm will first look for a solution $(S_u,S_v)$ where $S_u$ is independent. 
We say that it is doing a {\it $u$-feasibility check}.
If afterwards we have not found a solution $(S_u,S_v)$ where $S_u$ is independent, then our algorithm will repeat the same steps but now under the assumption that the set~$S_v$ is independent. That is, in that case our algorithm will perform a {\it $v$-feasibility check}.

\bigskip
\noindent
{\bf Phase 3: Doing a $\mathbf{u}$-feasibility check}

\medskip
\noindent
We start by exploring the structure of a solution $(S_u,S_v)$ that is either single-sided or independent, and where $S_u$ is an independent set. As $S_u$ and $N(u)$ are both independent sets,  $G[N(u)\cup S_u]$ is a connected bipartite graph.  Hence, $S_u$ contains a set~$S_u^*$, such that $S_u^*$ covers $N(u)$. We assume that $S_u^*$ has minimum size. Then each $s\in S_u^*$ has a nonempty set $Q(s)$ of neighbours in $N(u)$ that are not adjacent to any vertex in $S_u^*\setminus \{s\}$; otherwise we can remove $s$ from $S_u^*$, contradicting our assumption that $S_u^*$ has minimum size. We call the vertices of $Q(s)$ the {\it private} neighbours of $s$ with respect to $S_u^*$.

We note that $N(u)\cup S_u^*$ does not have to be connected. However, as $(G,u,v)$ has no $7$-constant solution, and thus no $1$-constant solution, we find that $S_u^*$ has size at least~2. We may assume that there is no vertex $t\in S_u\setminus S_u^*$, such that $N(t)\cap N(u)$ strictly contains $N(s)\cap N(u)$ for some $s\in S_u^*$ (otherwise we put $t$ in $S_u^*$ instead of~$s$).    Let $Q_u$ be the union of all private neighbour sets~$Q(s)$ ($s\in S_u^*$).  As $|S_u^*|\geq 2$, we observe that $G[Q_u\cup S_u^*]$ is the disjoint union of a set of at least two stars whose centers belong to~$S_u^*$. 

First suppose that $N(u)\setminus Q_u=\emptyset$. As $N(u)\cup S_u$ is connected and $G[Q_u\cup S_u^*]$ is the disjoint union of at least two stars, there exists a vertex~$t\in S_u\setminus S_u^*$ that is adjacent to vertices $z\in Q(s)$ and $z'\in Q(s')$ for  two distinct vertices $s,s'\in S_u^*$. As $N(u)\setminus Q_u=\emptyset$, we find that $Q(s)=N(s)\cap N(u)$. By our choice of~$S_u^*$, this means that $Q(s)$ contains at least one vertex~$w$ that is not adjacent to $t$. Similarly, $Q(s')$ contains a vertex~$w'$ that is not adjacent to $t$. By the definition of $Q(s)$ and $Q(s')$, we find that  $w$ and $z$ are not adjacent to $s'$, and $w'$ is not adjacent to $s$. Then $\{w',s'\}\cup \{w,s,z,t\}$ is an induced $P_2+P_4$ of $G$, a contradiction.

From the above we find that $N(u)\setminus Q_u\neq \emptyset$. Let $y\in N(u)\setminus Q_u$. As $S_u^*$ covers $N(u)$ and $y\notin Q_u$, we find that $y$ must be adjacent to at least two vertices~$s,s'\in S_u^*$. Suppose $y$ is not adjacent to some vertex $s^*\in S_u^*$. Let  $z\in Q(s)$ and $z^*\in Q(s^*)$. By the definition of the sets~$Q(s)$ and~$Q(s^*)$, we find that $z$ is not adjacent to $s'$ and $s^*$ and that $z^*$ is not adjacent to $s$ and $s'$. In particular it holds that $z\neq z^*$. Then $\{s^*,z^*\}\cup \{z,s,y,s'\}$ is an induced $P_2+P_4$ in $G$, a contradiction. Hence,  $y$ must be adjacent to all of $S_u^*$, that is, $N(u)\setminus Q_u$ must be complete to $S_u^*$. Note that this implies that $N(u)\cup S_u^*$ is connected.

To summarize, if $(G,u,v)$ has a solution $(S_u,S_v)$ in which $S_u$ is an independent set, then the following holds for such a solution $(S_u,S_v)$:

\begin{itemize}
\item [(P)] The set~$S_u$ contains a subset~$S_u^*$ of size at least~2 that covers $N(u)$, such that each vertex in $S_u^*$ has a nonempty set~$Q(s)$ of private neighbours with respect to $S_u^*$, and moreover, the set $N(u)\setminus Q_u$, where $Q_u=\bigcup_{s\in S_u^*}Q(s)$, is nonempty and complete to $S_u^*$.
\end{itemize}

\noindent
\emph{Remark.} We emphasize that $S_u^*$ is unknown to the algorithm, as we constructed it from the unknown $S_u$, and consequently, our algorithm does not know (yet) the sets $Q(s)$. 

\bigskip
\noindent
{\bf Phase 3a: Reducing $\mathbf{N(u)\setminus Q_u}$ to a single vertex $\mathbf{w_u}$}

\medskip
\noindent
We will now branch into a polynomial number of smaller instances,
in which $N(u)\setminus Q_u$ consists of just one single vertex~$w_u$. As we will show below, we can even identify $w_u$ and $Q_u$ for each of these new instances. Again, we will ensure that if one of these new instances has a solution, then $(G,u,v)$ has as solution. If none of these new instances has a solution, then $(G,u,v)$ may still have a solution $(S_u,S_v)$. However, in that 
case~$S_u$ is not an independent set, while $S_v$ must be an independent set. As mentioned, we will check this by doing a $v$-feasibility check as soon as we have finished the $u$-feasibility check.

\medskip
\noindent
{\bf Branching II} ($O(n^4)$ branches)\\
We will determine exactly those vertices of $N(u)$ that belong to $Q_u$ via some branching, under the assumption that 
$(G,u,v)$ has a solution $(S_u,S_v)$, where $S_u$ is independent, that satisfies (P).
By (P), $S_u^*$ consists of at least two (non-adjacent) vertices $s$ and $s'$. Let $w\in Q(s)$ and $w'\in Q(s')$. We branch by considering all possible choices of choosing these four vertices. This leads to $O(n^4)$ branches, which we each process in the way described below.

If we selected $s$ and $s'$ correctly, then $s,s'$ belong to an independent set $S_u$ that together with $S_v=T\setminus S_u$ forms a solution for $(G,u,v)$ that is not $7$-constant. This implies that $\{s,s'\}$ does not cover $N(u)$.
Hence, we can pick a vertex $w^*\in N(u)\setminus \{w,w'\}$. 
If $w^*$ is adjacent to both $s$ and $s'$, then $w^*$ must belong to $N(u)\setminus Q_u$. 
In the other case, that is, if $w^*$ is adjacent to at most one of $s,s'$, then $w^*$ must belong to $Q_u$. 
Hence, we have identified in polynomial time the (potential) sets $Q_u$ and $N(u)\setminus Q_u$. Moreover, by applying the 
{\tt Contraction Rule} on $N(u)\cup \{s,s'\}$ we can contract $s$ and $s'$ away.
This also contracts all of $N(u)\setminus Q_u$ into a single vertex which, as we mentioned above, we denote by ~$w_u$. Thus $w_u$ is complete to $S_u^*$. 

We denote the resulting instance by $(G,u,v)$ again. We also
let $T_1=N(w_u)\cap T$ and $T_2=T\setminus T_1$. Note that $S_u^*\subseteq T_1$. As $S_u^*$ covers $N(u)$ and every vertex of $S_u^*$ is adjacent to $w_u$, we find that $N(u) \cup S_u^*$ is connected. 
Due to the latter and because every vertex of $T_2$ is not in $S_u^*$, we may put without loss of generality  every vertex $t\in T_2$ with a neighbour in $N(v)$ in $S_v$. That is, we may contract such a vertex~$t$ away by applying the {\tt Contraction Rule} on $N(v)\cup \{t\}$. By the same reason, we may contract every edge between two vertices in $T_2$.
Hence, we have proven the following claim.

\clm{\label{c-t2}$T_2$ is an independent set that is anticomplete to $N(v)$.}\\[-8pt]
\noindent
Note that by definition, no vertex of $T_2$ is adjacent to $w_u$ either. In a later stage we will modify $T_2$ and this property may no longer hold. However, we will always maintain the properties that $T_2$ is independent and anticomplete to $N(v)$.

By Lemma~\ref{l-constant} we check in polynomial time if $(G,u,v)$ has a $7$-constant solution. If so, then we are done.
From now on suppose that $(G,u,v)$ has no $7$-constant solution. Recall that we are still looking for a single-sided or independent solution $(S_u,S_v)$, where $S_u$ is an independent set.
We first show that we can modify $G$ in polynomial time such that afterwards $G[T]$ is $(K_3+P_1)$-free.

Suppose $G[T]$ contains an induced $K_3+P_1$, say with vertices $x_1,x_2,x_3,y$ and edges $x_1x_2$, $x_2x_3$, $x_3x_1$. Consider a solution $(S_u,S_v)$ for $(G,u,v)$, where $S_u$ is an independent set.
Recall that we already checked on $7$-constant solutions. Hence, $(S_u,S_v)$ is not $7$-constant. 
As $S_u$ is an independent set, at least two of $x_1,x_2,x_3$, say $x_1,x_2$, must belong to $S_v$.
Then $(S_v \cap (N[x_1]\cup N[x_2]))\cup N(v)$ is connected; otherwise, as $S_v\cup N(v)$ is connected by definition, there would exist a vertex 
$t\in S_v\setminus  (N[x_1]\cup N[x_2])$  with a neighbour in $N(v)$ that is not adjacent to $x_1$ and $x_2$, contradicting Claim~\ref{c-either}.
 As $y$ does not belong to $N[x_1]\cup N[x_2]$, this means that $y$ is not needed for $S_v$.

From the above we can do as follows.
If $y$ has a neighbour in $N(u)$, then we contract~$y$ away by applying the {\tt Contraction Rule} on $N(u)\cup \{y\}$. 
Otherwise, if~$y$ has no neighbour in $N(u)$, then $y\in T_2$. 
As $N(u)\cup S_u^*$ is connected for some set $S^*_u \subseteq S_u\cap T_1$, this means
that~$y$ is not needed for $S_u$ either.
Hence, we may contract the edge between $y$ and an arbitrary neighbour of $y$ (as $G$ is connected, $y$ has at least one such neighbour).
We apply this rule, in polynomial time, for every induced copy of $K_3+P_1$ in $G[T]$.
Note that Claim~\ref{c-t2} is maintained and that in the end the following claim holds.

\clm{\label{c-t1}$G[T]$ is $(K_3+P_1)$-free.}\\[-8pt]
\noindent
We will now do some further branching to obtain $O(n)$ smaller instances in which $G[T_1]$ is $K_3$-free, such that 
the following holds. 
If one of these new instances has a solution, then $(G,u,v)$ has as solution. If none of these new instances has a solution, then $(G,u,v)$ may still have a solution $(S_u,S_v)$, but in that case $S_u$ is not an independent set while $S_v$ must be an independent set; this will be verified when we do the $v$-feasibility check.

\medskip
\noindent
{\bf Branching III} ($O(n)$ branches)\\
We consider all possibilities of putting one vertex~$t\in T_1$ in $S_u$. This leads to $O(n)$ branches.
For each branch we do as follows.
As $t$ is adjacent to $w_u$ (because $t\in T_1$), we can contract $t$ away using the {\tt Contraction Rule} on $N(u)\cup \{t\}$. As $S_u$ is independent, every neighbour~$t'$ of~$t$ in~$T_1$ must go to~$S_v$. 
If such a neighbour~$t'$ is adjacent to a vertex of $N(v)$, this means that we may 
contract $t'$ away by using the {\tt Contraction Rule} 
on $N(v)\cup \{t'\}$. 
If $t'$ has no neighbour in $N(v)$, then we put $t' $ in $T_2$. By the {\tt Contraction Rule} we may contract all edges between
$t'$ and its neighbours in $T_2$, 
such that $T_2$ is an independent set again that is anticomplete to $N(v)$, so Claim~\ref{c-t2} is still valid (but $T_2$ may now contain vertices adjacent to $w_u$). We denote the resulting instance by $(G,u,v)$ again. 
As $G[T]$, and consequently, $G[T_1]$ is $(K_3+P_1)$-free due to Claim~\ref{c-t1}, we find afterwards that the following holds for each branch.

\clm{\label{c-t1b}$G[T_1]$ is $K_3$-free.}\\[-8pt]
\noindent
By Lemma~\ref{l-constant} we check in polynomial time if $(G,u,v)$ has an $7$-constant solution. If so, then we are done.
From now on assume that $(G,u,v)$ has no $7$-constant solution. 
Note that $(G,u,v)$ has no double-sided solution either, as then the original instance has a double-sided solution, which we already ruled out (alternatively, apply Claim~\ref{c-double}).
We will focus on the following task (recall that a solution $(S_u,S_v)$ for $(G,u,v)$ is independent if both $S_u$ and $S_v$ are independent sets).

\bigskip
\noindent
{\bf Phase 3b: Looking for independent solutions after branching} 

\medskip
\noindent
We will now branch to $O(n^5)$ smaller instances for which the goal is to find an independent solution.
As before, if one of the newly created instances has a solution, then $(G,u,v)$ has as solution. If none of these new instances has a solution, then $(G,u,v)$ may still have a solution $(S_u,S_v)$. However, in that case $S_u$ is not independent and $S_v$ must be an independent set. This will be verified when we do the $v$-feasibility check.

We say that an instance $(G,u,v)$ satisfies the $(*)$-property if the following holds:

\medskip
\noindent
$(*)$ If $(G,u,v)$ has a solution $(S_u,S_v)$ where $S_u$ is an independent set, then $(G,u,v)$ has an independent solution.

\medskip
\noindent
Let $D_1,\ldots,D_q$ be the connected components of $G[T]$ for some $q\geq 1$. 
First suppose that every $D_i$ consists of a single vertex. Then $G[T]$ is an independent set. Hence, any solution for $(G,u,v)$ will be independent. We conclude that $(*)$ holds already.
Now suppose that at least one of $D_1,\ldots,D_q$, say $D_1$, has more than one vertex.

We first consider the case where another $D_i$, say $D_2$, also has more than one vertex. 
We claim that $(*)$ is again satisfied already. In order to see this, assume that $(G,u,v)$ has 
a solution $(S_u,S_v)$, where $S_u$ is an independent set, but $S_v$ contains two adjacent vertices $x_1$ and $x_2$.
We assume without loss of generality that $x_1$ and $x_2$ belong to $D_2$. Hence, $V(D_1)$ is anticomplete to $\{x_1,x_2\}$.
Suppose $D_1\cap S_v\neq \emptyset$. Let $t\in D_1\cap S_v$.
Recall that $(S_u,S_v)$ is not a $7$-constant solution, as $(G,u,v)$ does not have such solutions.
Then, by Claim~\ref{c-either}, we find that $S_v$ has a set $S_v'$ that contains $x_1,x_2$ but not $t$, such that every vertex of $N(v)$ is adjacent to a vertex of $S_v'$ and $S_v'\cup N(v)$ is connected. Hence, we may put $t$ into $S_u$. 
Similarly, we may put every other vertex of $D_1\cap S_v$ into $S_u$. As $T_2$ is an independent set by Claim~\ref{c-t2}, at least one vertex of $D_1$ belongs to $T_1$ and is thus adjacent to $w_u\in N(v)$ by definition. This means that
$(S_u\cup (V(D_1)\cap S_v),S_v\setminus V(D_1))$ is another solution for $(G,u,v)$. However, this solution is double-sided, a contradiction. So, from now on, we assume that $D_1$ contains more than one vertex and that $D_2,\ldots,D_q$ each have exactly one vertex.

Recall that $T_2$ is an independent set that is anticomplete to $N(v)$ due to Claim~\ref{c-t2}. Suppose $t\in T_2$ does not belong
to $D_1$. Then $t$ is an isolated vertex of $G[T]$ that is not adjacent to any vertex of $N(v)$. As $G$ is connected, $t$ is adjacent to at least one vertex of $N(u)$. We apply the {\tt Contraction Rule} on $N(u)\cup \{t\}$ to contract $t$ away. Afterwards, we find that every vertex of $T_2$ must belong to $D_1$.

Let $B_1,\ldots,B_p$ be the connected components of $G[T_1\cap V(D_1)]$ for some $p\geq 1$. By Claim~\ref{c-p4free}, $G[T]$, and thus $G[T_1\cap V(D_1)]$, is $P_4$-free (note that we only contracted edges during the branching and thus maintained $P_4$-freeness due to Lemma~\ref{l-contract}).
As $G[T_1]$ is also $K_3$-free by Claim~\ref{c-t1}, each $B_i$ is a complete bipartite graph on one or more vertices due to Lemma~\ref{l-p4}. 

First suppose that $p=1$. Recall that $T_2$ is an independent set by Claim~\ref{c-t2} that belongs to $D_1$.
In this case we show how to branch into $O(n^2)$ new and smaller instances, such that $(G,u,v))$ has a solution $(S_u,S_v)$, in which $S_u$ is an independent set, if and only if one of these new instances has such a solution. Moreover, each new instance will have the property that either $(*)$ has been obtained or $p\geq 2$ holds.

\medskip
\noindent
{\bf Branching IV} ($O(n^2)$ branches)\\
We consider each possibility of choosing one vertex $t \in B_1$ to be placed in $S_u$. This leads to $O(n)$ branches. In each 
branch we contract $t$ away by the {\tt Contraction Rule} on $N(u)\cup \{t\}$ (note that $tw_u\in E(G)$, as $t\in T_1$). 
Since $S_u$ is an independent set, we must place all neighbours of $t$ in $S_v$. In order to contract these neighbours away using 
the {\tt Contraction Rule} we may need to branch once more by considering every vertex that is in 
$B_1$ and that has at least one neighbour in $N(v)$. This leads to $O(n)$ additional branches. Hence, the total number of branches for this stage is $O(n^2)$.
For each branch we observe that $T_1$ has become an independent set (as the vertices in the components $D_2,\ldots, D_q$ form an independent set as well).
By applying the {\tt Contraction Rule} we ensure that $T_2$ is an independent set that 
remains anticomplete to $N(v)$.

First suppose that $T_1\cap V(D_1)$ consists of a single vertex $t^*$. If $T_2\neq \emptyset$, then we do as follows.
Recall that we are looking for a solution $(S_u,S_v)$ with $T_2\subseteq S_v$. As $N(v)\cup S_v$ must be connected but $T_2$ is anticomplete to $N(v)$ by Claim~\ref{c-t2}, vertex~$t^*$ must be placed into $S_v$. Hence, if $t^*$ is not adjacent to a vertex in $N(v)$, we discard the branch. Otherwise we contract $T_2\cup \{t^*\}$ away by applying the
{\tt Contraction Rule} on $N(v)\cup T_2\cup \{t^*\}$. Hence, we obtained $T_2=\emptyset$. As $T_1$ is an independent set, this means that $(*)$ holds.

Now suppose that $T_1\cap V(D_1)$ consists of more than one vertex. As $T_1$ is an independent set, this means that $G[T_1\cap V(D_1)]$ has $p\geq 2$ 
connected components. Hence, we have arrived in the case where $p\geq 2$. We denote the resulting instance by $(G,u,v,)$ again, and we let also $B_1,\ldots,B_p$ denote the connected components of $G[T_1\cap V(D_1)]$ again.

\medskip
\noindent
From the above we are now in the situation where $(G,u,v)$ is an instance for which $p\geq 2$ holds. By Lemma~\ref{l-p4} and because $D_1$ is connected and $P_4$-free, $D_1$ has a spanning complete bipartite graph~$B^*$. As $p\geq 2$, all vertices of $V(B_1)\cup \cdots \cup V(B_p)$ belong to the same partition class of $B^*$. By definition, these vertices are in $T_1$. Hence, as $T_2$ is an independent set in $D_1$, all vertices of $T_2$ form the other bipartition class of $B^*$.
Consequently, $T_2$ is complete to $T_1\cap V(D_1)$. We will do some branching.

\medskip
\noindent
{\bf Branching V}  ($O(n)$ branches)\\
Every vertex of $T_2$ will belong to $S_v$ in any solution $(S_u,S_v)$ where $S_u$ is an independent set, but without having any 
neighbours in $N(v)$ due to Claim~\ref{c-t2}. This means that $S_v$ contains at least one vertex~$t$ of $V(D_1)\cap T_1$. 
We branch by considering all possibilities of choosing this vertex~$t$. 
Indeed, as $T_2$ is complete to $T_1$, it suffices to check single vertices $t\in T_1$ that have a neighbour in $N(v)$.
This leads to $O(n)$ branches.
For each branch we do as follows. We contract 
the vertices of $T_2\cup \{t\}$ away using the {\tt Contraction Rule} on $N(v)\cup T_2\cup \{t\}$. 
We denote the resulting instance by $(G,u,v)$ and observe that $T_2=\emptyset$, so $T=T_1$. 

\medskip
\noindent
Note that $G[T]=G[T_1]$ now consists of connected components $B_1',\ldots,B_{p'}'$ for some $p'\geq 1$, where each $B_i'$ is a complete bipartite graph.
If every $B_i'$ consists of a single vertex, then $G[T]$ is an independent set. Hence, any solution for $(G,u,v)$ will 
be independent. We conclude that $(*)$ holds.
Now suppose that at least one of $B_1',\ldots,B_{p'}'$, say $B_1'$, has more than one vertex.
If another $B_i'$, say $B_2'$, also has more than one vertex, then  $(*)$ is also satisfied already. 
We can show this in the same way as before, namely when we proved this for the sets $D_1,\ldots,D_q$. 
From now on we may assume that $B_1'$ consists of more than one vertex and that $B_2',\ldots,B_{p'}'$ have only one vertex.
So, in particular, $B_1'$ is a complete bipartite graph on at least two vertices. We will do some branching.

\medskip
\noindent
{\bf Branching VI} ($O(n^2)$ branches)\\
We consider each possibility of choosing one vertex $t \in B_1'$ to be placed in $S_u$. This leads to $O(n)$ branches. In each 
branch we contract $t$ away by applying the {\tt Contraction Rule} on $N(u)\cup \{t\}$ (note that $tw_u\in E(G)$, as $t\in T_1$, so we can indeed do this). 
Since $S_u$ is an independent set, we must place all neighbours of $t$ in $S_v$. In order to contract these neighbours away using 
the {\tt Contraction Rule} we proceed as follows. All neighbours of $t$ that are adjacent to $N(v)$ we can contract away by applying the
{\tt Contraction Rule} on $N(v)\cup \{t\}$. If all neighbours of $t$ disappeared this way, this yields $T=T_1$, an independent
set as required. Otherwise, we need to include another vertex of $B'_1$ into $S_v$. So we branch on the $O(n)$ vertices $t' \in B_1'$
that are adjacent to $v$. Contracting such a $t'$ away makes all other neighbours of $t$ adjacent to $N(v)$ and we can contract
them away. In any case, eventually we will end up with $T=T_1$ being independent.
Consequently, $S_v$ must be an independent set for any solution $(S_u,S_v)$ where $S_u$ is an independent set. This means  that 
we achieved $(*)$.

\medskip
\noindent
If we have not yet found a solution, then by achieving $(*)$, as shown above, we have reduced the problem to $O(n^5)$ instances, for which we search for an independent solution. We consider these new instances one by one. For simplicity,
we denote the instance under consideration by $(G,u,v)$ again.

\bigskip
\noindent
{\bf Phase 3c: Searching for private solutions}

\medskip
\noindent
In this phase we introduce a new type of independent solution that we call private. In order to define this notion, we first describe our branching procedure which will get us to this new notion. 

\medskip
\noindent
{\bf Branching VII.} ($O(n^4)$ branches)\\
First we process $N(v)$ in the same way as we did for $N(u)$ in Branching II. That is, in polynomial time via $O(n^4)$ branches, we find a partition of $N(v)$ into a set~$Q_v$ of private neighbours and a vertex~$w_v$ that will be complete to $S_v$. To be more specific, if $(G,u,v)$ has a solution $(S_u,S_v)$ in which $S_u$ and $S_v$ are independent sets, then the following holds for such a solution 
$(S_u,S_v)$:

\begin{itemize}
\item [(P1)] The independent set~$S_u$ contains a subset~$S_u^*$ of size at least~2 that covers $N(u)$, such that each $s\in S_u^*$ has a nonempty set $Q_u(s)$ of private neighbours with respect to $S_u^*$, and moreover, the set $N(u)\setminus Q_u$, where
$Q_u=\bigcup Q_u(s)$, consists of a single vertex~$w_u$ that is complete to $S_u^*$.
\item [(P2)] The independent set~$S_v$ contains a subset~$S_v^*$ of size at least~2 that covers $N(v)$, such that each $s\in S_v^*$ has a nonempty set $Q_v(s)$ of private neighbours with respect to $S_v^*$, and moreover, the set $N(v)\setminus Q_v$, where $Q_v=\bigcup Q_v(s)$, consists of a single vertex~$w_v$ that is complete to $S_v^*$.
\end{itemize}

\noindent
We call an independent solution $(S_u,S_v)$ satisfying (P1) and (P2) a \emph{private solution}. We emphasize that by now all branches are guaranteed to have private solutions
or no solutions at all. Thus in what follows we will only search for private solutions. While doing this we may modify the instance $(G,u,v)$, but we will always ensure that private solutions are pertained.
In particular, if we contract a vertex $t \in S_u^*$ to $w_u$ using the {\tt Contraction Rule} on $N(u)\cup \{t\}$, then this leads to a private solution $(S_u,S_v)$ with $t \notin S_u^*$.
Then all private neighbours of $t$ become adjacent to $w_u$ and, by the 
{\tt Contraction Rule}, they get contracted to $w_u$ as well. However, if $t \notin S_u^*$, then contracting $t$ to $w_u$ will  make 
the neighbours of $t$ in $N(u)$ adjacent to $w_u$ and the {\tt Contraction Rule} contracts these to $w_u$. As a consequence, 
some vertices in $S_u^*$ may have no private neighbours in $N(u)$ and hence leave $S_u^*$. If this reduces $|S_u^*|$ to $1$, then we will notice this by checking for $1$-constant solutions, which takes polynomial time due to Lemma~\ref{l-constant}. If we find a $1$-constant solution, then we stop and conclude that our original instance is a yes-instance. Otherwise, we know that $|S_u^*| \ge 2$, and hence private solutions pertain (should such solutions exist at all). In the remainder, we will perform this test implicitly whenever we apply the {\tt Contraction Rule}.

We now prove the following two claims.
 
\clm{\label{c-both}Every vertex of $T$ is adjacent to both $w_u$ and $w_v$.}\\[-8pt]
\emph{Proof of Claim~\ref{c-both}.} 
Consider a vertex $t\in T$. 
First suppose that $t\in T$ is neither adjacent to $w_u$ nor to $w_v$.
As $G$ is connected, $t$ will be adjacent to some other vertex in $S_u$ or $S_v$ in every solution $(S_u,S_v)$.
Hence, $(G,u,v)$ has no independent solutions, and thus no private solutions, and we can discard the branch.
From now on assume that every vertex in $T$ is adjacent to at least one of $w_u$, $w_v$.
If $t\in T$ is adjacent to only $w_u$ and not to $w_v$, then by the same argument we must apply the {\tt Contraction Rule} on
$N(u)\cup \{t\}$. Similarly, if $t\in T$ is adjacent to only $w_v$ and not to $w_u$, then we must apply the {\tt Contraction Rule} on
$N(v)\cup \{t\}$.
We discard a branch whenever two adjacent vertices in $T$ were involved in an edge contraction with some neighbour in $N(u)$, or with some neighbour in $N(v)$. \dia
 
\clm{\label{c-bipartite} If $(G,u,v)$ has a private solution, then $G[T]$ must be the disjoint union of one or more complete bipartite graphs.}\\[-8pt]
\noindent
{\it Proof of Claim~\ref{c-bipartite}.} If $G[T]$ is not bipartite, then $(G,u,v)$ has no independent solution $(S_u,S_v)$, as
$T=S_u\cup S_v$. Hence, $(G,u,v)$ has no private solution. Assume that $G[T]$ is bipartite.
By Claim~\ref{c-p4free}, $G[T]$ is $P_4$-free. Then the claim follows by Lemma~\ref{l-p4}.\dia

\medskip
\noindent
By Claim~\ref{c-bipartite} we may assume that $G[T]$ is the disjoint union of one or more complete bipartite graphs; otherwise we discard the branch.

We now prove that $T$ can be changed into an independent set via some branching.
Suppose $T$ is not an independent set yet.
Let $B_1,\ldots,B_r$, for some $r\geq 1$, denote the connected components of $G[T]$ that have at least one edge (note that 
$G[T]$ may also contain some isolated vertices).
By Claim~\ref{c-bipartite}, every $B_i$ is complete bipartite. 

\clm{\label{c-4} If $(G,u,v)$ has a private solution, then $r\leq 3$.}\\[-8pt]
\noindent
\emph{Proof of Claim~\ref{c-4}.}
Assume that $r\geq 4$. We will prove that $(G,u,v)$ has no private solution.
Suppose that $T$ contains four connected components with edges, say $B_1,\ldots, B_4$, for which the following holds:
$V(B_1)$ covers some subset $A_1^u\subseteq N(u)$ and $V(B_2)$ covers some subset $A_2^u\subseteq N(u)$, such that
$A_1^u\setminus A_2^u \neq \emptyset$, whereas 
$V(B_3)$ covers some subset $A_3^v\subseteq N(v)$ and $V(B_4)$ covers some subset $A_4^v\subseteq N(v)$, such that
$A_3^v\setminus A_4^v \neq \emptyset$. 
Let $w\in A_1^u\setminus A_2^u$, say $w$ is adjacent to vertex $s$ of $B_1$ (and not to any vertex of $B_2$). Let $x_1$ and $x_2$ be two adjacent vertices of $B_2$, which exist as $B_2$ contains an edge.
Suppose $V(B_1)\cup V(B_2)$ does not cover $N(u)$. Then there exists a vertex $w'$ that has no neighbour in $V(B_1)\cup V(B_2)$. However, then $\{x_1,x_2\}\cup \{s,w,u,w'\}$ is an induced $P_2+P_4$ of $G$, a contradiction. Hence,
$V(B_1) \cup V(B_2)$ covers $N(u)$. Similarly, we find that $V(B_3)\cup V(B_4)$ covers $N(v)$. Then, as each vertex of $T$ is adjacent to both $w_u$ and $w_v$, we find that $G[N(u)\cup V(B_1)\cup V(B_2)]$ and $G[N(u)\cup V(B_3)\cup V(B_4)]$ are connected.
This is not possible, as then the original instance has a double-sided solution, which we already ruled out after Claim~\ref{c-double}.

If two sets from $V(B_1),\ldots,V(B_r)$, say $V(B_1)$ and $V(B_2)$, cover the same subset of $N(u)$ \emph{and} the same subset of $N(v)$, then we can apply the {\tt Contraction Rule} on $N(u)\cup V(B_1)$ and on $N(v)\cup V(B_2)$ to
find that the original instance has a double-sided solution if it has a solution. However, as we already ruled this out, this
is not possible either.

Now consider the sets $B_1$ and $B_2$. From the above, we deduce the following. We may assume without loss of generality that 
$V(B_1)$ and $V(B_2)$ cover different subsets of~$N(u)$. This implies that 
$V(B_3),\ldots, V(B_r)$ all cover the same subset $A$ of $N(v)$. We can also apply the above on $B_1$ and $B_3$ to find that $B_1$ and $B_3$ must either cover different subsets of $N(u)$ or different subsets of $N(v)$. Suppose $B_1$ and $B_3$ cover different subsets of~$N(v)$. Then, again from the above, $B_2, B_4,\ldots, B_r$ must cover the same subset of $N(u)$. As $B_1$ and $B_2$ cover different subsets of~$N(u)$, this means that $B_1$ and $B_4$ cover different subsets of $N(u)$. This implies that $B_2$ must cover the same 
set $A$ as $B_4,\ldots,B_r$. As $B_3$ covers $A$ as well, this means that $B_2$ and $B_3$ cover the same subset of $N(v)$. Hence, they must cover different subsets of $N(u)$. However, the latter implies that $B_1$ and $B_4$ cover the same subset of $N(v)$. As $B_4$ covers $A$, just like $B_3$, we find that $B_1$ and $B_3$ cover the same subset $A$ of $N(v)$, a contradiction. Hence, $B_1$ and $B_3$ cover the same subset of $N(v)$, namely $A$, and by symmetry the same holds for $B_2$. 
 
As $V(B_1)$ covers $A_1^u$ and $V(B_2)$ covers $A_2^u$ such that
$A_1^u\setminus A_2^u \neq \emptyset$, we can use the same arguments as before to deduce 
that $V(B_1)$ and $V(B_2)$ must cover $N(u)$. 
We put the vertices of $B_1$ and $B_2$ into $S_u$ and the vertices of
$B_i$ for $i\geq 3$ plus all other (isolated) vertices of $T$ into $S_v$. 
If $(S_u,S_v)$ is a solution for $(G,u,v)$, then the original solution has a double-sided solution, which we already ruled out. Hence, there exists a vertex $z$ of $N(v)$ that is not adjacent to any vertex
of $T\setminus (V(B_1)\cup V(B_2))$. However, as every $V(B_i)$ covers the same subset $A$ of $N(v)$, no vertex of $B_1$ and $B_2$ is adjacent to $z$ either. This implies that $(G,u,v)$ is a no-instance, meaning that we can discard this branch.\dia

\medskip
\noindent
By Claim~\ref{c-4} we may assume that $r\leq 3$, that is, $G[T]$ has at most three connected components $B_i$ with an edge;
otherwise we discard the branch.
As $r\leq 3$, we can now do some  branching to obtain $O(1)$
 smaller instances in which $T$ is an independent set, such 
that $(G,u,v)$ has a private solution if and only if at least one of these new instances has a private solution.

\medskip
\noindent
{\bf Branching VIII} ($O(1)$ branches)\\
For $i=1,\ldots, r$, let $Y_i$ and $Z_i$ be the bipartition classes of $B_i$. 
Let $1\leq i\leq r$. As $S_u$ and $S_v$ must be independent sets and every $B_i$ is complete bipartite, either $Y_i$ belongs to $S_u$ and $Z_i$ belongs to $S_v$, or the other way around. We branch by considering both possibilities. We do this for
each $i\in \{1,\ldots, r\}$. This leads to $2^r\leq 2^3$ branches, as $r\leq 3$ due to Claim~\ref{c-4}. In each branch we apply the {\tt Contraction Rule} to contract
$Y_i$ and $Z_i$ away (note that here our remark about pertaining private solutions applies).
We consider every resulting instance separately. We denote such an instance again by $(G,u,v)$, for which we have proven the following claim.

\clm{\label{c-indep} $T$ is an independent set.}\\[-8pt]
\noindent
We now continue as follows. As $T$ is an independent set by Claim~\ref{c-indep}, the sets~$S_u$ and~$S_v$ of any solution 
$(S_u,S_v)$ will  be independent (should $(G,u,v)$ have a solution). Recall also that $\{w_u,w_v\}$ is complete to $T$ by Claim~\ref{c-both}. We are looking for a 
private
solution $(S_u,S_v)$, which we recall is an independent solution for which sets $S_u^*$ and $S_v^*$ exist so that (P1) and (P2) are satisfied.
We make the following observation. Let $R=T\setminus (S_u^*\cup S_v^*)$ be the set of all other vertices of $T$. Consider a vertex $z\in R$. We note that if $z\in S_u$, then $(S_u\setminus \{z\},S_v\cup \{z\})$ is also a solution for $(G,u,v)$; this follows from (P1) and (P2) and the fact that
$w_v$ is adjacent to $z\in T$.
Similarly, if $z\in S_v$, then $(S_u\cup \{z\},S_v\setminus \{z\})$ is a solution as well.

We prove the following four claims.

\clm{\label{c-atleast} Let $w\in N(u)\cup N(v)$. Then we may assume without loss of generality that $w$ is adjacent to at least two vertices of $T$.}\\[-8pt]
\noindent
{\it Proof of Claim~\ref{c-atleast}.} Suppose $N(u)$ or $N(v)$, say $N(u)$, contains a vertex~$w$ that is adjacent to at most 
one vertex of $T$. If $w$ has no neighbours in $T$, then $(G,u,v)$ has no solution and we discard the branch. Suppose $w$ has
exactly one neighbour $z\in T$. Then $z$ belongs to 
$S_u^*$ for every (private) solution $(S_u,S_v)$ of $(G,u,v)$ (assuming $(G,u,v)$ is a yes-instance). Hence, we may apply the 
{\tt Contraction Rule} on $N(u)\cup \{z\}$. We apply this operation exhaustively, while pertaining private solutions as before.

It may happen that in this process it turns out that two vertices $z,z'$ both belong to 
$S_u^*$ for every (private) solution $(S_u,S_v)$ of $(G,u,v)$,
while they share a neighbour in $N(u)\setminus \{w_u\}$. This contradicts (P1). Hence, in this case we find that
$(G,u,v)$ does not have a private solution and we may discard the branch.
Otherwise, in the end, we have obtained in polynomial time an instance with the desired property.
As we ensure that private solutions pertain, the size of $S_u^*$ remains at least $2$. \dia 

\clm{\label{c-notall} Let $z\in T$. Then we may assume without loss of generality that $z$ is non-adjacent to at least one vertex of $N(u)$ and to at least one vertex of $N(v)$.}\\[-8pt]
\noindent
{\it Proof of Claim~\ref{c-notall}.} Suppose $z\in T$ is adjacent to all vertices of $N(u)$ or to all vertices of $N(v)$, say 
to all vertices of $N(u)$. Then we can check in polynomial time if $T\setminus \{z\}$ covers $N(v)$. If so, then 
$\{z,T\setminus \{z\}$ is a ($1$-constant) solution of $(G,u,v)$ and we can stop. Otherwise $z$ must belong to $S_v$ for any 
solution $(S_u,S_v)$ of $(G,u,v)$. 
In that case we may apply the {\tt Contraction Rule} on 
$N(v)\cup \{z\}$. We apply this operation exhaustively (we again recall that we ensure that private solutions pertain by checking for $1$-constant solutions).
Moreover, it may happen that during this process two vertices $z,z'$ will end up 
in the same set~$S_u$ or~$S_v$ for any private solution $(S_u,S_v)$, while sharing a neighbour in $N(u)\setminus \{w_u\}$. 
As the sets~$S_u$ and~$S_v$ are independent in a private solution $(S_u,S_v)$, this means 
that $(G,u,v)$ does not have a private solution and we may discard the branch.
Otherwise, in the end, we have obtained in polynomial time an instance with the desired property.\dia

\clm{\label{c-three} 
Let $s$ and $t$ be any two distinct vertices of $T$. Then we may assume without loss of generality that either $N(u)\cap N(s)\cap N(t)=\{w_u\}$; or $N(u)\cap N(s)=N(u)\cap N(t)$; or $\{s,t\}$ covers $N(u)$. Similarly, we may assume without loss of generality that either $N(v)\cap N(s)\cap N(t)=\{w_v\}$; or $N(v)\cap N(s)=N(v)\cap N(t)$; $\{s,t\}$ covers $N(v)$.}\\[-8pt]
\noindent
{\it Proof of Claim~\ref{c-three}.} By symmetry it suffices to prove only the first statement. Assume $T$ contains two vertices $s$ and $t$, for which there exist distinct vertices $w\in (N(u)\setminus \{w_u\})\cap N(s)\cap N(t)$; $w'\in (N(u)\cap N(s))\setminus N(t)$ and $w''\in N(u)\setminus (N(s)\cup N(t))$. Note that $w_u\notin \{w,w',w''\}$. 
Recall that in this stage we are looking for private solutions for $(G,u,v)$.
Consider an arbitrary private solution $(S_u,S_v)$ (if it exists). Then $w''\in Q_u(z)$ for some $z\in S_u^*$. Note that $z\notin \{s,t\}$, as neither $s$ nor $t$ is adjacent to $w''$.

The above means that $z$ must be adjacent to at least one of $w,w'$, as otherwise the set $\{w'',z\}\cup \{w',s,w,t\}$ induces a $P_2+P_4$ in $G$, which is not possible. Hence, at least one of $w$ or $w'$ will be a private neighbour of $z$, that is, will belong to $Q_u(z)$. As $s$ is adjacent to both $w$ and $w'$ and $N(u)\setminus Q_u=\{w_u\}$ (see property~(P1) of the definition of a private solution), this means that $s$ does not belong to $S_u^*$. 
We conclude that $s$ belongs to $R=T\setminus (S_u^*\cup S_v^*)$ or to $S_v^*$ for any private solution $(S_u,S_v)$ of $(G,u,v)$.
As $s$ is adjacent to $w_v\in N(v)$, we may therefore apply the {\tt Contraction Rule} on $N(v)\cup \{s\}$, 
ensuring persistence of private solutions (in case there are any) in the usual way. 
We do this exhaustively, and in the end we find that the claim holds. Note that we obtained this situation in polynomial time. \dia

\clm{\label{c-wu} Let $s$ and $t$ be any two distinct vertices of $T$ that together cover $N(u)$. 
Then 
there exists a nonempty set 
$A(v)\subseteq N(v)$ that is complete to $\{s,t\}$ and
anticomplete to $T\setminus \{s,t\}$, or  $(G,u,v)$ has a $2$-constant solution. The same holds for $u$ and $v$ interchanged.}\\[-8pt] 
{\it Proof of Claim~\ref{c-wu}.} Assume without loss of generality that $\{s,t\}$ covers $N(u)$. Then we find that $(\{s,t\}, T\setminus \{s,t\})$ is a $2$-constant solution unless 
$N(v)$ contains a nonempty set $A(v)$ that is anticomplete to $T\setminus \{s,t\}$).
By Claim~\ref{c-atleast} we find that $A(v)$ is complete to $\{s,t\}$.
\dia

\medskip
\noindent
We will use Claims~\ref{c-atleast}--\ref{c-wu} to prove the following claim.

\clm{\label{c-notboth} Let $s$ and $t$ be two distinct vertices in $T$ such that $\{s,t\}$ covers 
$N(u) \cup N(v)$.
Then $(G,u,v)$ has a $2$-constant solution.}\\[-8pt]
\noindent
{\it Proof of Claim~\ref{c-notboth}.} 
Assume that $(G,u,v)$ has no $2$-constant solution. Then by Claim~\ref{c-wu}, there is a nonempty set $A(u)\subseteq N(u)$ that is complete to $\{s,t\}$ and anticomplete to $T\setminus \{s,t\}$. 
Similarly, there exists a nonempty set $A(v)\subseteq N(v)$ that is complete to $\{s,t\}$ and anticomplete to $T\setminus \{s,t\}$.
By Claim~\ref{c-notall} we find that $N(u)$ contains a vertex~$w$ that is not adjacent to $s$. As $\{s,t\}$ covers $N(u)$, this means that $w$ is adjacent to $t$. By Claim~\ref{c-atleast} we find that $w$ is adjacent to some vertex $s'\in T\setminus \{s,t\}$.
As $s'$ is anticomplete to $A(u)$, Claim~\ref{c-three} tells us that $\{s',t\}$ covers $N(u)$.
By the same argument, there exists a vertex $t'$ such that $\{s,t'\}$ covers $N(v)$. Putting $s',t$ in $S_u$ and $s,t'$ in $S_v$ (together with all other vertices of $T$) yields a $2$-constant solution $(S_u,S_v)$ of $(G,u,v)$. This is a contradiction. \dia

\medskip
\noindent
We continue as follows. By Lemma~\ref{l-constant} we check in polynomial time if $(G,u,v)$ has a $2$-constant solution. If so, then we are done. Otherwise, we obtain the following claim, which immediately follows from Claim~\ref{c-notboth} and the fact that
if one pair of vertices of $T$ covers $N(u)$ and another pair covers $N(v)$, then we obtained a $2$-constant solution.

\clm{\label{c-onlyr} We may assume without loss of generality that every pair of (distinct) vertices $\{s,t\}$ in $T$ does not cover $N(u)$; hence, $\{s,t\}$ may only cover $N(v)$.}\\[-8pt]
We call a pair of vertices $s,t$ of $T$ a {\it $2$-pair} if $\{s,t\}$ covers $N(v)$.
Let $T_v$ be the set of vertices of $T$ involved in a 2-pair. We continue by proving the following claim.

\clm{\label{c-2pair} Every vertex of $T_v$ belongs to exactly one $2$-pair.}\\[-8pt]
\noindent
{\it Proof of Claim~\ref{c-2pair}.} Let $s\in T_v$. By definition, $s$ belongs to at least one $2$-pair. 
For contradiction, suppose that $s$ belongs to more than one $2$-pair.
Then there exist vertices $t_1$, $t_2$ in $T_v$, such that $\{s,t_1\}$ and $\{s,t_2\}$ both cover $N(v)$.
As $(G,u,v)$ has no $2$-constant solution, $N(u)$ contains a nonempty set $A_1(u)\subseteq N(u)$ that is complete to $\{s,t_1\}$ and 
anticomplete to $T\setminus \{s,t_1\}$ due to Claim~\ref{c-wu}. By the same claim, $N(u)$ contains a nonempty set $A_2(u)\subseteq N(u)$ that is 
complete to $\{s,t_2\}$ and anticomplete to $T\setminus \{s,t_2\}$.
Let $w_1\in A_1(u)$ and $w_2\in A_2(u)$; note that $w_2\neq w_u$. Then $w_1$ is adjacent to $s$ but not to $t_2$, whereas $w_2\neq w_u$ is a common neighbour of $s$ and $t_2$. As $\{s,t_2\}$ does not cover $N(u)$ due to Claim~\ref{c-onlyr}, this contradicts Claim~\ref{c-three}. \dia

\medskip
\noindent
We next prove that actually $T_v = \emptyset$.
Suppose that $T_v\neq \emptyset$. 
Let $(s,t) \in T_v$. By Claim~\ref{c-wu}, there exists a nonempty subset $A(u)$ of $N(u)$ that is complete to $\{s,t\}$
and anticomplete to $T\setminus \{s,t\}$. As $(G,u,v)$ has no $2$-constant solution, $s$ and $t$ do not cover all of $N(u)$.  
By Claim~\ref{c-notall}, we find that $s$ is not adjacent to some vertex
$w\in N(v)$. As $(s,t)$ is a 2-pair, $t$ is adjacent to $w$.
By Claim~\ref{c-atleast}, we find that $w$ is adjacent to a vertex $z\in T\setminus \{s,t\}$. From Claim~\ref{c-2pair} 
it follows that $(t,z)$ is not a 2-pair, so $t$ and $z$ do not cover all of $N(v)$. 
By Claim~\ref{c-three} and the fact that $t$ and $z$ have a common neighbour
different from $w_v$, namely $w$, this means that $t$ and $z$ are adjacent to the same neighbours in $N(v)$. However, then $(s,z)$ is 2-pair, 
contradicting Claim~\ref{c-2pair}.
This means that we have indeed proven the following claim.

\clm{\label{c-tempty} $T_v=\emptyset$.}\\[-8pt]

\noindent
{\bf Phase 3d: Translating the problem into a matching problem}

\medskip
\noindent
We are now ready to translate the instance $(G,u,v)$ into an instance of a matching problem.
Recall that $w_u$ and $w_v$ are the vertices in $N(u)$ and $N(v)$ that are complete to $T$.
By Claims~\ref{c-three} and~\ref{c-tempty} we can partition $N(u)\setminus \{w_u\}$ into sets $N_1(u)\cup \dots \cup N_q(u)$ 
for some $q\geq 1$ such that two vertices of $N(u)$ have the same set of neighbours in $T$ if and only if they both belong to $N_i(u)$ for some
$i\in \{1,\ldots,q\}$. Similarly, we can partition $N(v)\setminus \{w_v\}$ into sets $N_1(v)\cup \dots \cup N_r(v)$ for some $r\geq 1$ such that two vertices of $N(v)$ have the same set of neighbours in $T$ if and only if they both belong to $N_i(v)$ for some
$i\in \{1,\ldots,r\}$.
We may remove all but one vertex of each $N_h(u)$ and each $N_i(v)$ to obtain an equivalent instance, which we denote by $(G,u,v)$ again. 

Let $G'$ be the graph obtained from $G$ after removing the vertices $u,v,w_u,w_v$ and every edge between a vertex of $N(u)$ and a vertex of $N(v)$. Note that $G'$ is bipartite with partition classes $(N(u)\setminus \{w_u\})\cup (N(v)\setminus \{w_v\})$ and $T$. It remains to compute a maximum matching~$M$ in~$G'$. We can do this by using the
Hopcroft-Karp algorithm, which runs in $O(m\sqrt{n})$-time on bipartite graphs with $n$ vertices and $m$ edges.
If $|M|=|N(u)|+|N(v)|-2$, then each vertex in $(N(u)\setminus \{w_u\})\cup (N(v)\setminus \{w_v\})$ is incident to an edge of $M$, and hence, we found a (private) solution for $(G,u,v)$. If $|M|<|N(u)|+|N(v)|-2$, then $(G,u,v)$ has no (private) solution, and we discard the branch.

\medskip
\noindent
The above concludes the description of the $u$-feasibility check. If we found a branch with a solution, then
we translate it in polynomial time to a solution for the original instance. Otherwise we perform Phase 4.

\bigskip
\noindent
{\bf Phase 4: Doing a $\mathbf{v}$-feasibility check}

\medskip
\noindent
 As mentioned, our algorithm now does a $v$-feasibility check, that is, it checks for the existence of a solution $(S_u,S_v)$, where $S_v$ is an independent set and $G[S_u]$ may contain edges. As we can repeat exactly the same steps as in Phase~3, this phase takes polynomial time as well. This concludes the description of our algorithm.

\medskip
\noindent
The correctness of our algorithm follows from the above description. We now analyze its run-time. The branching is done in eight 
stages, namely Branching I-VIII and yields a total number of $O(n^{30})$ branches. As explained in each step above, processing
each branch created in Branching I-VI  until we start branching again takes polynomial time. 
Checking for $1$-constant solutions to ensure survival of private solutions takes constant time as well. Moreover, processing each of the branches created in Branch VII takes polynomial time as well. 
We conclude that the total running time of our algorithm is polynomial.  \qed
\end{proof}

Via Lemma~\ref{l-reduce} and a reduction to $P_4$-{\sc Suitability} we obtain:

\begin{lemma}\label{l-top5}\hspace*{-0.1cm}$P_5$-{\sc Suitability} can be solved in polynomial time for $(P_2+P_4)$-free graphs.
\end{lemma}

\begin{proof}
Let $(G,u,v)$ be an instance of $P_5$-{\sc Suitability}, where $G$ is a connected $(P_2+P_4)$-free graph. We may assume without loss of generality that $u$ and $v$ are of distance at least~4 from each other, as otherwise $(G,u,v)$ is a no-instance. By the {\tt Contraction Rule} and Lemma~\ref{l-contract} we may also assume without loss of generality that $N(u)$ and $N(v)$ are both independent sets; otherwise if, say, $G[N(u)]$ contains an edge~$e$, then we contract $e$ to obtain an equivalent but smaller instance $(G',u,v)$, where $G'$ is also $(P_2+P_4)$-free due to Lemma~\ref{l-contract}. 

First suppose $|N(u)|=1$, say $N(u)=\{u'\}$ for some $u'\in V(G)$. Then we solve $P_4$-{\sc Suitability} on instance $(G-u,u',v)$. We can do this in polynomial time due to Lemma~\ref{l-top4}.
 
Now suppose $|N(u)|\geq 2$. Note that $\dist(u,v) \leq 5$, as $G$ is $P_7$-free. By Lemma~\ref{l-reduce} we may assume that $\dist(u,v)=4$. We will explore the structure of the $P_5$-witness bags $W(p_2)$ and $W(p_3)$ should they exist. Let $Z$ be the set that consists of all vertices $z$ with $\dist(u,z)=\dist(z,v)=2$. Then $Z$ must be a subset of $W(p_3)$. As $N(u)$ is not connected, $W(p_2)$ must contain at least one other vertex~$s$ adjacent to some vertex $t \in N(u)$. Suppose $s$ is non-adjacent to some other vertex $t'\in N(u)$. Let $w$ be a neighbour of $v$. As $s \in W(p_2)$ and $w \in W(p_4)$, we find that $s$ and $w$ are not adjacent. Then the set $\{v,w\}\cup \{s,t,u,t'\}$ induces a $P_2+P_4$ in $G$, a contradiction. Hence, $s$ is adjacent to every vertex of $N(u)$. We consider all possibilities of choosing vertex $s$ from the set $V(G)\setminus (N[u]\cup N[v]\cup Z)$. This leads to $O(n)$ branches. In each branch we contract the set $N(u)\cup \{s\}$ to a single vertex $u'$. Let $G'$ be the resulting graph. Then we solve $P_4$-{\sc Suitability} on instance $(G',u',v)$. As~$G'$ is $(P_2+P_4)$-free by Lemma~\ref{l-contract}, we can do this in polynomial time due to Lemma~\ref{l-top4}. 

From the above we conclude that we can check in polynomial time if $(u,v)$ is a $P_5$-suitable pair of~$G$. \qed
\end{proof}

We use Lemma~\ref{l-top5} to prove Lemma~\ref{l-top6}.

\begin{lemma}\label{l-top6}\hspace*{-0.1cm}$P_6$-{\sc Suitability} can be solved in polynomial time for $(P_2+P_4)$-free graphs.
\end{lemma}

\begin{proof}
Let $(G,u,v)$ be an instance of $P_6$-{\sc Suitability}, where $G$ is a connected $(P_2+P_4)$-free graph. We may assume without loss of generality that $u$ and $v$ are of distance at least~5 from each other, as otherwise $(G,u,v)$ is a no-instance. We may also assume without loss of generality that $N(u)$ and $N(v)$ are both independent sets; otherwise if, say, $G[N(u)]$ contains an edge~$e$, then we contract $e$ to obtain an equivalent but smaller instance $(G',u,v)$, where $G'$ is also $(P_2+P_4)$-free due to Lemma~\ref{l-contract}. 

First suppose $|N(u)|=1$, say $N(u)=\{u'\}$ for some $u'\in V(G)$. Then we solve $P_5$-{\sc Suitability} on instance $(G-u,u',v)$. We can do this in polynomial time due to Lemma~\ref{l-top5}. 

Now suppose $|N(u)|\geq 2$.  We assume $W(p_1)=\{u\}$ and we will explore the structure of the $P_6$-witness bag $W(p_2)$ should it exist. As $N(u)$ is not connected, $W(p_2)$ must contain at least one other vertex~$s$. Suppose that $s$ is adjacent to some vertex $t\in N(u)$ and non-adjacent to some other vertex $t'\in N(u)$. Let $w$ be a neighbour of $v$. Then the set $\{v,w\}\cup \{s,t,u,t'\}$ induces a $P_2+P_4$ in $G$, a contradiction. Hence, $s$ is adjacent to every vertex of $N(u)$. We consider all possibilities of choosing vertex $s$ from the set $V(G)\setminus (N[u]\cup N[v])$. This leads to $O(n)$ branches. In each branch we contract the set $N(u)\cup \{s\}$ to a single vertex $u'$. Let $G'$ be the resulting graph. Then we solve $P_5$-{\sc Suitability} on instance $(G',u',v)$. As~$G'$ is $(P_2+P_4)$-free by Lemma~\ref{l-contract}, we can do this in polynomial time due to Lemma~\ref{l-top5}. 

From the above we conclude that we can check in polynomial time if $(u,v)$ is a $P_6$-suitable pair of $G$. \qed
\end{proof}

We now combine Lemmas~\ref{l-outer} and~\ref{l-trivial} with Lemmas~\ref{l-top4}--\ref{l-top6} to obtain the following theorem.

\begin{theorem}\label{t-p2p4}
The {\sc Longest Path Contractibility} problem is polynomial-time solvable for $(P_2+P_4)$-free graphs.
\end{theorem}

\begin{proof}
Let $G$ be a connected $(P_2+P_4)$-free graph. We may assume without loss of generality that $G$ has at least one edge. Note that $G$ is $P_7$-free. Hence, $G$ does not contain $P_7$ as a a contraction. By combining Lemmas~\ref{l-top4}--\ref{l-top6} with Lemma~\ref{l-outer} we can check in polynomial time if $G$ contains $P_k$ as a contraction for $k=6,5,4$. If not, then we check if $G$ contains $P_3$ as a contraction by using Lemma~\ref{l-trivial} combined with Lemma~\ref{l-outer}.  If not then, as $G$ has an edge, $P_2$ is the longest path to which $G$ can be contracted to. \qed
\end{proof}

\subsection{The Case $\mathbf{H=P_1+P_2+P_3}$}\label{s-p1p2p3}

We will prove that {\sc Longest Path Contractibility} is polynomial-time solvable for $(P_1+P_2+P_3)$-free graphs.

We will start by showing that $P_4$-{\sc Suitability} is polynomial-time solvable for $(P_1+P_2+P_3)$-free graphs.
The proof of this result uses similar but more simple arguments than the proof of Lemma~\ref{l-top4}.

\begin{lemma}\label{l-top4c}
The $P_4$-{\sc Suitability} problem can be solved in polynomial time for $(P_1+P_2+P_3)$-free graphs.
\end{lemma}

\begin{proof}
\setcounter{ctrclaim}{0}
Let $(G,u,v)$ be an instance of $P_4$-{\sc Suitability}, where $G$ is a connected $(P_1+P_2+P_3)$-free graph. We may assume without loss of generality that $u$ and $v$ are of distance at least 3, that is, $u$ and $v$ are non-adjacent and $N(u)\cap N(v)=\emptyset$; otherwise $(G,u,v)$ is a no-instance. 

Recall that $T=V(G)\setminus (N[u]\cup N[v])$. Recall also that we are looking for a partition $(S_u,S_v)$ of~$T$ that is a solution for $(G,u,v)$, that is,  $N(u)\cup S_u$ and $N(v)\cup S_v$ must both be connected.  In order to do so we will construct partial solutions $(S_u',S_v')$, which we try  to extend to a solution $(S_u,S_v)$ for $(G,u,v)$. We use the {\tt Contraction Rule} from Section~\ref{s-pre} on $S_u'$ and $S_v'$, so that these sets will become independent. By Lemma~\ref{l-contract}, the resulting graph will always be $(P_1+P_2+P_3)$-free.  For simplicity, we denote the resulting instance by $(G,u,v)$ again. After applying the {\tt Contraction Rule} the size of the set $T$ may be reduced by at least one. As before, if $t\in T$ was involved in an edge contraction with a vertex from $N(u)$ or $N(v)$ when applying the rule, then we say that we contracted $t$ away.
 
We start by applying the {\tt Contraction Rule} on $N(u)$ and $N(v)$. This leads to the following claim.

\clm{\label{c-ind2} $N(u)$ and $N(v)$ are independent sets.}\\[-8pt]
\noindent
We now check if $(G,u,v)$ has an $8$-constant solution, which we can check in polynomial time due to Lemma~\ref{l-constant}. If so, then $(G,u,v)$ is a yes-answer and we stop. From now on suppose  that $(G,u,v)$ has no $8$-constant solution. Then we prove the following claim (recall that a solution $(S_u,S_v)$ is independent if $S_u$ and $S_v$ are independent sets).

\clm{\label{c-ind3} Every solution of $(G,u,v)$ is independent (if $(G,u,v)$ has solutions).}\\[-8pt]
{\it Proof of Claim~\ref{c-ind3}.} Let $(S_u,S_v)$ be a solution for $(G,u,v)$ that is not independent, say $s,t$ belong to $S_u$ with 
$st\in E(G)$. If $\{s,t\}$ is anticomplete to a set of two neighbours $w,w'$ of $u$, then $\{v\}\cup \{s,t\}\cup \{w,u,w'\}$ is an induced $P_1+P_2+P_3$ of $G$, a contradiction. Hence, $\{s,t\}$ covers all but at most one vertex of $N(u)$. Suppose that $\{s,t\}$ covers $N(u)$, Then, as $s$ and $t$ are adjacent in $G$, we find that $(S_u,S_v)$ is a $2$-constant solution and thus a $8$-constant solution, which is not possible. Hence, $N(u)$ contains a unique vertex $w$ that is not adjacent to $s$ and $t$, but that is adjacent to some $z\in T\setminus \{s,t\}$.  As $G$ is $(P_1+P_2+P_3)$-free, $G$ is $P_8$-free. Then $G[N(u)\cup S_u]$ contains a path~$P$ on at most seven vertices from $s$  to $z$. The path $P$, together with vertex~$t$ that may not be on $P$, shows that $(S_u,S_v)$ is  a $8$-constant solution, a contradiction. \dia

\medskip
\noindent
We will now analyze the structure of an independent solution $(S_u,S_v)$. As $S_u$ and $N(u)$ are both independent sets, $G[N(u)\cup S_u]$ is a connected bipartite graph. Hence, $S_u$ contains a set $S_u^*$, such that $S_u^*$ covers $N(u)$. We assume that $S_u^*$ has minimum size. Then each $s\in S_u^*$ has a nonempty set $Q(s)$ of vertices in $N(u)$ that are not adjacent to any vertex in $S_u^*\setminus \{s\}$; otherwise we can remove $s$ from $S_u^*$, contradicting our assumption that $S_u^*$ has minimum size. We call the vertices of $Q(s)$ the {\it private} neighbours of $s\in S_u^*$ with respect to $S_u^*$.

As $(G,u,v)$ has no $8$-constant solution, and thus no $1$-constant solution, we find that $S_u^*$ has size at least~2.
Suppose $Q(s)$ contains at least two private neighbours $w_1,w_2$ of some vertex $s\in S_u^*$.
As $|S_u^*|\geq 2$, there exists a vertex $s'\in S_u^*$ with $s'\neq s$. Let $w_3\in Q(s')$. Then 
$\{v\}\cup \{w_3,s'\}\cup \{w_1,s,w_2\}$ is an induced $P_1+P_2+P_3$ of $G$, a contradiction.
Hence, each set $Q(s)$ has size~1. We denote the unique vertex of $Q(s)$ by $w_u^s$. So, $w_u^s$ is adjacent to $s$ but not to any other vertex from $S_u^*$.
Let $Q_u$ be the set of all vertices $w_u^s$. Then $G[Q_u\cup S_u^*]$ is the disjoint union of 
$|S_u^*|$ edges. 

We claim that the set~$N(u)\setminus Q_u$ is complete to $S_u^*$.
In order to see this, let $w\in N(u)\setminus Q_u$.  By definition, $w$ is adjacent to at least two vertices $s_1,s_2$ of $S_u^*$. 
For contradiction, assume that $w$ is not adjacent to some vertex $s_3\in S_u^*$. Then $\{v\} \cup \{s_3,w_u^{s_3}\}\cup \{s_1,w,s_2\}$ induces a $P_1+P_2+P_3$ in $G$, which is not possible. 

As $G[N(u)\cup S_u]$ is connected as well and $S_u$ is an independent set, every vertex $t\in S_u\setminus S_u^*$ must be adjacent to at least one vertex of $N(u)$.
However, we claim that every vertex of $S_u\setminus S_u^*$ is adjacent to at most one vertex of~$Q_u$.
For contradiction, assume that $S_u\setminus S_u^*$ contains a vertex~$t$ that is adjacent to two vertices of $Q_u$, say to $w_u^s$ and $w_u^{s'}$ for some $s,s'\in S_u^*$ with $s\neq s'$. Recall that $S_u$ is independent. Consequently,
if $t$ is non-adjacent to $w_u^{s''}$ for some $s''\in S_u^*\setminus \{s,s'\}$, then $G$ contains an induced $P_1+P_2+P_3$ with
vertex set $\{v\} \cup \{w_u^{s''},s''\} \cup \{w_u^s,t,w_u^{s'}\}$, a contradiction. Hence, $t$ is adjacent to every vertex of $Q_u$.
If $t$ is adjacent to every vertex of $N(u)$, then $(G,u,v)$ has a $1$-constant solution, and this an $8$-constant solution, which we ruled out already. Hence, the set $N(u)\setminus N(t)$ is nonempty. 
As $t$ is adjacent to every vertex of $Q_u$, the set $N(u)\setminus N(t)$ is a subset of $N(u)\setminus Q_u$.
Recall that $N(u)\setminus Q_u$ is complete to $S_u^*$. Hence, $N(u)\setminus N(t)$ is complete to $S_u^*$.
Let $s\in S_u^*$. Then $\{s,t\}$ covers $N(u)$, and moreover $G[N(u)\cup \{s,t\}]$ is connected. This means that $(G,u,v)$ has a $2$-constant solution and thus an $8$-constant solution, which is not possible. 
We conclude that every vertex of $S_u\setminus S_u^*$ is adjacent to at most one vertex of~$Q_u$.

Finally, we prove that $N(u)\setminus Q_u$ is nonempty. For contradiction, assume that $N(u)\setminus Q_u$ is empty. Then 
$N(u)=Q_u$. As $G[Q_u\cup S_u^*]$ is the disjoint union of a number of edges, and $G[N(u)\cup S_u]$ is connected, there must exist a vertex $t\in S_u\setminus S_u^*$ that is adjacent to at least two vertices of $Q_u$. However, we proved above that this is not possible. We conclude that $N(u)\setminus Q_u$ is nonempty.

We can deduce all the claims above with respect to $v$ as well. To summarize, any independent solution $(S_u,S_v)$ for $(G,u,v)$ satisfies the following two properties: 

\begin{itemize}
\item [(P1)] The independent set $S_u$ contains a subset $S_u^*$ of size at least~2 that covers $N(u)$, such that each vertex  $s\in S_u^*$ has exactly one private neighbour $w_u^s$ in $N(u)$ with respect to $S_u^*$, and moreover, the set $N(u)\setminus Q_u$, where $Q_u=\{w_u^s\; |\; s\in S_u^*\}$, is nonempty and complete to $S_u^*$, and every vertex of $S_u\setminus S_u^*$ is adjacent to at most one vertex of $Q_u$ and to at least one vertex of $N(u)\setminus Q_u$.
\item [(P2)] The independent set $S_v$ contains a subset $S_v^*$ of size at least~2 that covers $N(v)$, such that each vertex in $s\in S_v^*$ has exactly one private neighbour $w_v^s$ in $N(v)$ with respect to $S_v^*$, and moreover, the set $N(v)\setminus Q_v$, where $Q_v=\{w_v^s\; |\; s\in S_v^*\}$, is nonempty and complete to $S_v^*$, and  every vertex of $S_v\setminus S_v^*$ is adjacent to at most one vertex of $Q_v$ and to at least one vertex of $N(u)\setminus Q_v$.
\end{itemize}

\noindent
\emph{Remark.} We emphasize that $S_u^*$ and $S_v^*$ are unknown to the algorithm, as we constructed it from the unknown sets $S_u$ and $S_v$, and consequently our algorithm does not know (yet) the sets $Q_u$ and $Q_v$.  

\medskip
\noindent
We will now branch into $O(n^8)$ smaller instances
in which $N(u)\setminus Q_u$ and $N(u)\setminus Q_v$ consist of just one single vertex~$w_u$ and $w_v$, respectively, such that $(G,u,v)$ has an independent solution if and only if 
at least one of the new instances has an independent solution. Moreover, we will be able to identify $w_u$ and $w_v$, and consequently, the sets $Q_u$ and $Q_v$, in polynomial time.

\medskip
\noindent
{\bf Branching} ($O(n^8)$ branches)\\
We will determine exactly those vertices of $N(u)$ that belong to $Q_u$ via some branching, under the assumption that 
$(G,u,v)$ has an independent solution $(S_u,S_v)$ that satisfies (P1) and (P2).
By (P1), $S_u^*$ consists of at least two (non-adjacent) vertices $s$ and $s'$. 
By (P2), $S_v^*$ consists of at least two (non-adjacent) vertices $t$ and $t'$.
We branch by considering all possible choices of choosing these four vertices together with their private neighbours
$w_u^s$, $w_u^{s'}$, $w_v^{t}$, $w_v^{t'}$ (which are unique by (P1) and (P2)). This leads to $O(n^8)$ branches.

For each branch we do as follows. 
We discard the branch in which $G[\{s,s',w_u^s,w_u^{s'}\}]$ and $G[\{t,t',w_v^t,w_v^{t'}\}]$ are not both isomorphic to $2P_2$.
We put a vertex $y\in N(u)$ in $N(u)\setminus Q_u$ if and only if $y$ is a common neighbour of $s$ and $s'$. This gives us the set~$Q_u$. We obtain the set~$Q_v$ in the same way. If there exists a vertex in $Q_u\setminus \{w_u^s,w_u^{s'}\}$ that is adjacent to one of $s,s'$, then we discard the branch. We also discard the branch if there exists a vertex in $Q_v\setminus \{w_v^t,w_v^{t'}\}$ that is adjacent to one of $t,t'$.
Moreover, by applying the 
{\tt Contraction Rule} on $N(u)\cup \{s,s'\}$ we can contract $s$ and $s'$ away.
This contracts all vertices of $N(u)\setminus Q_u$ into a single vertex which we denote by~$w_u$ due to (P1). 
Similarly, we branch $t$ and $t'$ away and this leads to the contraction of $N(v)\setminus Q_v$ into a single vertex~$w_v$ due to (P2). Note that we have identified $w_u$ and $w_v$ in polynomial time. We denote the resulting instance by $(G,u,v)$ again. 

\medskip
\noindent
Consider a vertex $z\in T$. Firs suppose that $z$ is not adjacent to $w_u$. Then $z$ does not belong to~$S_u$ in any independent solution $(S_u,S_v)$ for $(G,u,v)$ by (P1).Hence $z$ must belong to $S_v$ for any independent solution $(S_u,S_v)$ for $(G,u,v)$. 
However, (P2) tells us that If $z$ is not adjacent to $w_v$, then $z$ cannot belong to the set~$S_v$ of any independent solution $(S_u,S_v)$ for $(G,u,v)$. Hence, in that case we must discard the branch. 
Otherwise, that is, if $z$ is adjacent to $w_v$, then we check the following. If $z$ has two neighbours in $N(v)\setminus \{w_v\}$, then $z$ does not belong to $S_v$ in any independent solution $(S_u,S_v)$ for $(G,u,v)$ due to (P2). Hence, we will discard the branch. 
If $z$ is adjacent to at most one vertex of $N_v\setminus \{w_v\}$, then we apply  the {\tt Contraction Rule} on $N(v)\cup \{z\}$ to contract $z$ away. As a side effect, the possible neighbour of $z$ in $N_v\setminus \{w_v\}$ will be contracted away as well.
Now suppose that $z$ is not adjacent to $w_v$. Then we perform the same operation with respect to $u$. 
We apply this operation exhaustively on both $u$ and $v$. This takes polynomial time. In the end we either discarded the branch or have found a new instance, which we also denote by $(G,u,v)$ again, in which every vertex of $T$ is adjacent to $w_u$ and to $w_v$.

Consider again a vertex $z\in T$.
If $z$ is adjacent to only $w_u$ and $w_v$ and 
to at most one other vertex~$w$ in $N(u)\cup N(v)$, then we apply the {\tt Contraction Rule} on $G[N(u)\cup \{z\}]$ (if $w\in N(u)$) or $G[N(v)\cup \{z\}]$ (in the other two cases) in order to contract $z$ away.
As a side effect, the possible other neighbour of $z$ in $(N(u)\cup N(v)) \setminus \{w_u,w_v\}$ will be contracted away as well.
If $z$ is adjacent to more than one vertex of $N(u)\setminus \{w_u\}$, then $z$ does not belong to $S_u$ in any independent solution $(S_u,S_v)$ for $(G,u,v)$. We check if $z$ is adjacent to more than one vertex of $N(v)\setminus \{w_v\}$. If so,
 then $z$ does not belong to $S_v$ in any independent solution $(S_u,S_v)$ for $(G,u,v)$. In that case we will discard the branch.
 Otherwise we will apply the {\tt Contraction Rule} on $N(v)\cup \{z\}$ to contract $z$ away. Again, as a side effect, the possible neighbour of $z$ in $N(v)\setminus \{w_v\}$ will be contracted away as well.
 If $z$ is adjacent to more than one vertex of $N(u)\setminus \{w_v\}$, we perform a similar operation with respect to $u$.
 We apply this rule exhaustively. This takes polynomial time. 
 In the end we find that every vertex of $T$ is adjacent to $w_u$ and $w_v$ and to exactly one vertex of $Q_u$ and to exactly one vertex of $Q_v$.

We now remove all edges of $G[T]$. We also remove $w_u$ and $w_v$ from the graph. This yields a bipartite graph $G'$ with partition classes $N(u)\cup N(v)\setminus \{w_u,w_v\}$ and~$T$. 
It remains to compute a maximum matching~$M$ in~$G'$. We can do this by using the
Hopcroft-Karp algorithm~\cite{HK73}, which runs in $O(m\sqrt{n})$-time on bipartite graphs with $n$ vertices and $m$ edges.
If $|M|=|N(u)|+|N(v)|-2$ then we found a solution for $(G,u,v)$; otherwise we discard the branch. Note that we did not explicitly forbid that two adjacent vertices of $T$ ended up in $S_u$ or two adjacent vertices of $T$ ended up in $S_v$: we have ruled out the existence of such solutions already (but they would still be perfectly acceptable if they did exist).

\medskip
\noindent
As mentioned, we translate a solution found for some branch into a solution for the original instance. We can do so in polynomial time.
If we find no yes-answer for the instance of any branch, then we conclude that the original instance has no solution. 

\medskip
\noindent
The correctness of our algorithm follows from the above description. We now analyze its run-time. 
There is only one branching procedure, which yields a total number of $O(n^{8})$ branches. As explained above, processing
each branch takes polynomial time. In particular, checking for $8$-constant solutions takes polynomial time due to Lemma~\ref{l-constant}.
We conclude that the total running time of our algorithm is polynomial. \qed
\end{proof}

We proceed in the same way as in the case where $H=P_2+P_4$. That is, we will use Lemma~\ref{l-top4c} to prove 
Lemma~\ref{l-top5c}. Then we use Lemma~\ref{l-top5c} to prove Lemma~\ref{l-top6c}, and we use Lemma~\ref{l-top6c} to prove
Lemma~\ref{l-top7c}.

\begin{lemma}\label{l-top5c}
The $P_5$-{\sc Suitability} problem can be solved in polynomial time for $(P_1+P_2+P_3)$-free graphs.
\end{lemma}

\begin{proof}
Let $(G,u,v)$ be an instance of $P_5$-{\sc Suitability}, where $G$ is a connected $(P_1+P_2+P_3)$-free graph. We may assume without loss of generality that $u$ and $v$ are of distance at least~4 from each other, as otherwise $(G,u,v)$ is a no-instance.  We may also assume without loss of generality that $N(u)$ and $N(v)$ are independent sets; otherwise, say $N(u)$ contains an edge, we apply the {\tt Contraction Rule} on $N(u)$ to obtain an equivalent but smaller instance $(G',u,v)$, where $G'$ is also $(P_1+P_2+P_3)$-free due to Lemma~\ref{l-contract}.

If $N(u)$ consists of exactly one vertex $u'$, then we can instead solve $P_5$-{\sc Suitability} on instance $(G-u,u',v)$. By Lemma~\ref{l-top5c} this takes polynomial time.  Hence, we may assume that $N(u)$, and for the same reason, $N(v)$ have size at least~2.

By Lemma \ref{l-reduce} we may assume that $\dist(u,v)=4$. Let $M$ consist of all vertices of $G$ that are of distance~2 from $u$ and of distance~2 from $v$. Note that $M\neq \emptyset$, as $\dist(u,v)=4$. Moreover, if $G$ has a $P_5$-witness structure ${\cal W}$ with $W(p_1)=\{u\}$ and $W(p_5)=\{v\}$, then $M\subseteq W(p_3)$ must hold. 

Let $z,z'$ be two vertices in $N(v)$. Suppose $x\notin M\cup \{u\}$ is adjacent to $w\in N(u)$ but not to $w'\in N(u)$. As $x$ is not in $M$ and adjacent to $w\in N(u)$, we find that $x$ is not adjacent to $z$ and $z'$. However, then $\{w'\}\cup \{w,x\}\cup \{z,v,z'\}$ induces a $P_1+P_2+P_3$ in $G$, a contradiction. Hence, every vertex not in $M\cup \{u\}$  is either complete to $N(u)$ or anticomplete to $N(u)$.  This means that if $G$ has a $P_5$-witness structure ${\cal W}$ with $W(p_1)=\{u\}$ and $W(p_5)=\{v\}$, then the following holds: $W(p_2)\setminus N(u)$ contains a vertex $s$, such that $N(u)\cup \{s\}$ is connected.

We now branch by considering all possibilities of choosing this vertex~$s$; note that we only have to consider vertices of $G$ that are of distance~2 from $u$ and that are not in $M$. This leads to $O(n)$ branches. We consider each branch separately, as follows. First we contract all edges in $G[N(u)\cup \{s\}]$. If this does
not
yield a single vertex $u'$, then we discard the branch. Otherwise we let $G'$ be the resulting graph. The graph $G'-u'$ consists of at least two connected components, one of which consists of vertex~$u$, and the other one contains $v$ and $N(v)$. We contract away the vertices of any other connected component~$D$ of $G'-u'$ by applying the {\tt Contraction Rule} on $\{u'\}\cup V(D)$. It remains to check if $(G'-u,u',v)$ is a yes-instance of $P_4$-{\sc Suitability}. We can do this in polynomial time via Lemma~\ref{l-top4c}. As there are $O(n)$ branches, the total running time of our algorithm is polynomial. \qed
\end{proof}

\begin{lemma}\label{l-top6c}
The $P_6$-{\sc Suitability} problem can be solved in polynomial time for $(P_1+P_2+P_3)$-free graphs.
\end{lemma}

\begin{proof}
Let $(G,u,v)$ be an instance of $P_6$-{\sc Suitability}, where $G$ is a connected $(P_1+P_2+P_3)$-free graph.
We may assume without loss of generality that $u$ and $v$ are of distance at least~5 from each other, as otherwise
$(G,u,v)$ is a no-instance. 
We may also assume without loss of generality that $N(u)$ and $N(v)$ are independent sets; otherwise, say $N(u)$ contains an edge, we apply the {\tt Contraction Rule} on $N(u)$ to obtain an equivalent but smaller instance $(G',u,v)$, where $G'$ is also $(P_1+P_2+P_3)$-free due to Lemma~\ref{l-contract}.

If $N(u)$ consist of exactly one vertex $u'$, then we can instead solve $P_5$-{\sc Suitability} on instance $(G-u,u',v)$. By Lemma~\ref{l-top5c} this takes polynomial time. 
Hence, we may assume that $N(u)$, and for the same reason, $N(v)$ are independent sets of size at least~2. Let $z,z'$ be two vertices in $N(v)$.
Suppose $x\notin N(u)\cup \{u\}$ is adjacent to $w\in N(u)$ but not to $w'\in N(u)$. Then
$\{w'\}\cup \{w,x\}\cup \{z,v,z'\}$ induces a $P_1+P_2+P_3$ in $G$, a contradiction.
Hence, every vertex not in $N(u)\cup \{u\}$ is either complete to $N(u)$ or anticomplete to $N(u)$. 
This means that if $G$ has a $P_5$-witness structure ${\cal W}$ with
$W(p_1)=\{u\}$ and $W(p_5)=\{v\}$, then the following holds:
$W(p_2)\setminus N(u)$ contains a vertex $s$, such that $N(u)\cup \{s\}$ is connected.

We now branch by considering all possibilities of choosing this vertex~$s$. This leads to $O(n)$ branches. We consider each branch separately, as follows.
First we contract all edges in $G[N(u)\cup \{s\}]$. If this does 
not
yield a single vertex $u'$, then we discard the branch. Otherwise we let $G'$ be the resulting graph. The graph $G'-u'$ consists of at least two connected components, one of which consists of vertex~$u$, and the other one contains $v$ and $N(v)$. We contract away the vertices of any other connected component~$D$ of $G'-u'$ by applying 
the {\tt Contraction Rule} on $\{u'\}\cup V(D)$. It remains to check if $(G'-u,u',v)$ is a yes-instance of $P_4$-{\sc Suitability}. We can do this in polynomial time via Lemma~\ref{l-top5c}. As there are $O(n)$ branches, the total running time of our algorithm is 
polynomial. \qed
\end{proof}

\begin{lemma}\label{l-top7c}
The $P_7$-{\sc Suitability} problem can be solved in polynomial time for $(P_1+P_2+P_3)$-free graphs.
\end{lemma}

\begin{proof}
Let $(G,u,v)$ be an instance of $P_7$-{\sc Suitability}, where $G$ is a connected $(P_1+P_2+P_3)$-free graph.
We may assume without loss of generality that $u$ and $v$ are of distance at least~6 from each other, as otherwise
$(G,u,v)$ is a no-instance. Note that in fact $u$ and $v$ are of distance exactly~6 from each other, as otherwise 
$G$ contains an induced $P_1+P_2+P_3$.
We may also assume without loss of generality that $N(u)$ is an independent set; otherwise
we apply the {\tt Contraction Rule} on $N(u)$ to obtain an equivalent but smaller instance $(G',u,v)$, where $G'$ is also $(P_1+P_2+P_3)$-free due to Lemma~\ref{l-contract}.

Suppose $N(u)$ contains two vertices $w$ and $w'$. As $u$ and $v$ are of distance~6 from each other, there exists a vertex~$y$
with $\dist(u,y)=\dist(v,y)=3$. Let $z\in N(v)$. Then the set $\{y\}\cup \{v,z\}\cup \{w,u,w'\}$ induces a $P_1+P_2+P_3$ in~$G$, a contradiction. Hence, $N(u)$ consist of exactly one vertex $u'$. 
We can therefore solve $P_6$-{\sc Suitability} on instance $(G-u,u',v)$. By Lemma~\ref{l-top6c} this takes polynomial time. \qed
\end{proof}

We are now ready to prove the main result of Section~\ref{s-p1p2p3}.

\begin{theorem}\label{t-p1p2p3}
The {\sc Longest Path Contractibility} problem is polynomial-time solvable for $(P_1+P_2+P_3)$-free graphs.
\end{theorem}

\begin{proof}
Let $G$ be a connected $(P_1+P_2+P_3)$-free graph.
We may assume without loss of generality that $G$ has at least one edge.
Then $G$ is $P_8$-free. Hence, $G$ does not contain $P_8$ as a a contraction.
By combining Lemmas~\ref{l-top4c}--\ref{l-top7c} with Lemma~\ref{l-outer} we can check in polynomial time if $G$ contains $P_k$ as a contraction for $k=7,6,5,4$. If not, then we check if $G$ contains $P_3$ as a contraction by using Lemma~\ref{l-trivial} combined with Lemma~\ref{l-outer}.  If not then, as $G$ has an edge, $P_2$ is the longest path to which $G$ can be contracted to. \qed
\end{proof}

\subsection{The Case $\mathbf{H=P_1+P_5}$}\label{s-p1p5}

We will prove that {\sc Longest Path Contractibility} is polynomial-time solvable for $(P_1+P_5)$-free graphs.
This result extends a corresponding result of~\cite{HPW09} for $P_5$-free graphs. Its proof is based on the same but slightly generalized arguments as the result for $P_5$-free graphs and comes down to the following lemma.

\begin{lemma}\label{l-all}
Let $k\geq 4$ and let $G$ be a $(P_1+P_5)$-free graph with a $P_k$-suitable pair $(u,v)$ such that $N(u)$ is an independent set.
Then $G$ has a $P_k$-witness structure ${\cal W}$ with $W(p_1)=\{u\}$ and $W(p_k)=\{v\}$, for which the following holds:
$W(p_2)\setminus N(u)$ contains a set~$S$ of size at most~$2$ such that $N(u)\cup S$ is connected.
\end{lemma}

\begin{proof}
As $(u,v)$ is a $P_k$-suitable pair, $G$ has a $P_k$-witness structure ${\cal W}$ with $W(p_1)=\{u\}$ and $W(p_k)=\{v\}$.
For contradiction, assume that $W(p_2)\setminus N(u)$ contains no set~$S$ of size at most~2 such that $N(u)\cup S$ is connected.
Then $W(p_2)\setminus N(u)$ contains at least three vertices $x_1$, $x_2$, $x_3$ such that one of the following holds:

\begin{itemize}
\item [(i)] for $i=1,2,3$, vertex $x_i$ is adjacent to some vertex $w_i\in N(u)$ with $w_i\notin N(x_h)\cup N(x_j)$, where $\{h,i,j\}=\{1,2,3\}$; or
\item [(ii)] $N(u)\subseteq N(x_1)\cup N(x_2)$, but $G[N(u)\cup \{x_1\}\cup \{x_2\}]$ is not connected.
\end{itemize}

First assume that~(i) holds. Recall that $N(u)$ is an independent set.
Then $x_1x_2 \in E(G)$, as otherwise the set $\{v\}\cup \{x_1,w_1,u,w_2,x_2\}$ induces a $P_1+P_5$ in $G$, which is not possible.
However, now the set $\{v\}\cup \{w_3,u,w_2,x_2,x_1\}$ induces a $P_1+P_5$ in $G$, a contradiction.
Hence, (i) cannot hold.
Now assume that (ii) holds.
As $(G,u,v)$ has no $1$-constant solution,  $x_1$ has a neighbour $w_1\in N(u)$ not adjacent to $x_2$ and
$x_2$ has a neighbour $w_2\in N(u)$ not adjacent to $x_1$. As $G[N(u)\cup \{x_1\}\cup \{x_2\}]$ is not connected but
$N(u)\subseteq N(x_1)\cup N(x_2)$,
we have that $x_1x_2\notin E(G)$ However, then
the set $\{v\}\cup \{x_1,w_1,u,w_2,x_2\}$ induces a $P_1+P_5$ in $G$, a contradiction. Hence~(ii) does not hold either, a contradiction. \qed
\end{proof}

As a consequence of Lemma~\ref{l-all}, we get that $P_4$-{\sc Suitability} is easy and that $P_k$-{\sc Suitability} reduces to $P_4$-{\sc Suitability}, as we will see.

\begin{lemma}\label{l-top4b}
The $P_4$-{\sc Suitability} problem can be solved in polynomial time for $(P_1+P_5)$-free graphs.
\end{lemma}

\begin{proof}
Let $(G,u,v)$ be an instance of $P_4$-{\sc Suitability}, where $G$ is a connected $(P_1+P_5)$-free graph.
We may assume without loss of generality that $u$ and $v$ are of distance at least~3 from each other, as otherwise
$(G,u,v)$ is a no-instance. 
We may also assume without loss of generality that $N(u)$ is an independent set; otherwise
we apply the {\tt Contraction Rule} on $N(u)$ to obtain an equivalent but smaller instance $(G',u,v)$, where $G'$ is also $(P_1+P_5)$-free due to Lemma~\ref{l-contract}. 
By Lemma~\ref{l-all} we find that if $(G,u,v)$ has a solution, then $G$ has a $2$-constant solution.
We can check the latter in $O(n^4)$ time by Lemma~\ref{l-constant}. \qed
\end{proof}

\begin{lemma}\label{l-top5b}
The $P_5$-{\sc Suitability} problem can be solved in $O(n^6)$ time for $(P_1+P_5)$-free graphs.
\end{lemma}

\begin{proof}
Let $(G,u,v)$ be an instance of $P_5$-{\sc Suitability}, where $G$ is a connected $(P_1+P_5)$-free graph.
We may assume without loss of generality that $u$ and $v$ are of distance at least~4 from each other, as otherwise
$(G,u,v)$ is a no-instance. 
We may also assume without loss of generality that $N(u)$ is an independent set; otherwise
we apply the {\tt Contraction Rule} on $N(u)$ to obtain an equivalent but smaller instance $(G',u,v)$, where $G'$ is also $(P_1+P_5)$-free due to Lemma~\ref{l-contract}.

If $(u,v)$ is a $P_5$-suitable pair, then by Lemma~\ref{l-all}, $G$ has a $P_5$-witness structure ${\cal W}$ with
$W(p_1)=\{u\}$ and $W(p_5)=\{v\}$, for which the following holds:
$W(p_2)\setminus N(u)$ contains a set~$S$ of size at most~$2$, such that $N(u)\cup S$ is connected.

We now branch by considering all possibilities of choosing this set $S$. This leads to $O(n^2)$ branches. We consider each branch separately, as follows.
First we contract all edges in $G[N(u)\cup S]$. If this does 
not
yield a single vertex $u'$, then we discard the branch. Otherwise we let $G'$ be the resulting graph. The graph $G'-u'$ consists of at least two connected components, one of which consists of vertex~$u$, and the other one contains $v$ and $N(v)$. 
If there are more components in $G'-u'$ than these two, we contract each such component $D$ to $u'$  by applying 
the {\tt Contraction Rule} on $\{u'\}\cup V(D)$. It remains to check if $(G'-u,u',v)$ is a yes-instance of $P_4$-{\sc Suitability}. We can do this in $O(n^4)$ time via 
Lemma~\ref{l-top4b}. 
As there are $O(n^2)$ branches, the total running time of our algorithm is $O(n^6)$.
\qed
\end{proof}

\begin{lemma}\label{l-top6b}
The $P_6$-{\sc Suitability} problem can be solved in $O(n^8)$ time for $(P_1+P_5)$-free graphs.
\end{lemma}

\begin{proof}
We reduce $P_6$-{\sc Suitability} to $P_5$-{\sc Suitability} in exactly the same way we reduced $P_5$-{\sc Suitability} to 
$P_4$-{\sc Suitability} in the proof of Lemma~\ref{l-top5b}. This leads to $O(n^2)$ branches. For each branch we apply Lemma~\ref{l-top5b}, which takes $O(n^6)$ time.
Hence the total running time of $O(n^8)$.
\qed
\end{proof}

We are now ready to prove the main result of Section~\ref{s-p1p5}.

\begin{theorem}\label{t-p1p5}
The {\sc Longest Path Contractibility} problem is polynomial-time solvable for $(P_1+P_5)$-free graphs.
\end{theorem}

\begin{proof}
Let $G$ be a connected $(P_1+P_5)$-free graph.
We may assume without loss of generality that $G$ has at least one edge. Note that $G$ is $P_7$-free. Hence, $G$ does not contain $P_7$ as a a contraction.
By combining Lemmas~\ref{l-top4b}--\ref{l-top6b} with Lemma~\ref{l-outer} we can check in polynomial time if $G$ contains $P_k$ as a contraction for $k=6,5,4$.
If not, then we check if $G$ contains $P_3$ as a contraction by using Lemma~\ref{l-trivial} combined with Lemma~\ref{l-outer}.  If not then, as $G$ has an edge, $P_2$ is the longest path to which $G$ can be contracted to. \qed
\end{proof}

\subsection{The Case $\mathbf{H=sP_1+P_4}$}\label{s-sp1p4}

We adopt/extend the notation from Section~\ref{s-p4}.
Let $(G,u,v)$ be an instance of $P_k$-{\sc Suitability} with $k \ge 4$. A \emph{solution} is a witness structure $\mathcal{W}=\{W(p_1), \dots,
W(p_k)\}$ with $W(p_1)=\{u\}, W(p_2)=N(u) \cup S_u , W(p_{k-1}) = N(v) \cup S_v$ and  $W(p_k)=\{v\}$. We let $T:= V\setminus
(N[u] \cup N[v])$. Thus $S_u$ and $S_v$ are disjoint subsets of $T$ such that $N(u) \cup S_u$ and $N(v) \cup S_v$ are connected.
As in Section~\ref{s-p4}, we call a solution $\alpha$-constant if there exists a subset $S_u'\subseteq S_u$
 with $N(u) \cup S_u'$ 
connected and $|S_u'| \le \alpha$, or there exists $S_v'\subseteq S_v$ with $N(v)\cup S_v'$ connected and $|S_v' |\le \alpha$.
 
Let $(\{u\}, N(u) \cup S_u, \dots)$ be a solution and $S'_u \subseteq S_u$ such that $N(u) \cup S'_u$ is connected. We define the \emph{closure} $\overline{S'_u}$ of $S'_u$ as the set of all
vertices in $S_u$ that are connected to $v$ in $G$ only via $N(u) \cup S'_u$. 
\begin{lemma}\label{l-clos}
Let $(\{u\},N(u) \cup S_u, W(p_3), \dots, W(p_{k-1}),\{v\})$ be a solution for an instance $(G,u,v)$ of $P_k$-{\sc Suitability} for some $k\geq 4$. 
If $S_u' \subseteq S_u$ such that $N(u) \cup S'_u$ is connected,
then $(\{u\},N(u) \cup \overline{S'_u}, W(p_3)\cup S_u \setminus \overline{S'_u}, \dots,W(p_{k-1}),\{v\})$ is also a solution for $(G,u,v)$.
\end{lemma}
\begin{proof}
We check the three properties for witness structures. All bags in the new partition are mutually disjoint. Connectedness of $W(p_3)\cup S_u \setminus \overline{S'_u}$: Any $s \in S_u\setminus\overline{S'_u}$ is
joined to $v$ by a path $P$ that does not pass through $N(u) \cup S'_u$ (by definition). Moreover, $P$ does not hit $\overline{S'_u}$ since from there; by definition, we cannot reach $v$ without passing through
$N(u) \cup S'_u$. Since vertices in $W(p_2)=N(u)\cup S_u$ are only adjacent to vertices in $W(p_1) \cup W(p_2) \cup W(p_3)$, we find that $P$ must be contained in $S_u \setminus \overline{S'_u}$ until it
eventually reaches $W(p_3)$. Connectedness of $W(p_3) \cup S_u \setminus \overline{S'_u}$ follows. \qed
\end{proof}

As it turns out, it suffices to search for $\alpha$-constant solutions:

 \begin{lemma}\label{l-constps} An instance $(G,u,v)$  of $P_k$-{\sc Suitability}, where $G$ is $(sP_1+P_4)$-free,
has a solution if and only if it has an $\alpha$-constant solution, where $\alpha=(s+2)(2s+4)$.
 \end{lemma}
\begin{proof} 
We may, as usual, assume that $N(u)$ is independent (otherwise we apply the {\tt Contraction Rule} to $N(u)$
without any effect on $S_u$ in solutions $(\{u\}, N(u) \cup S_u, \dots)$).
First suppose that $(G,u,v)$ has an $\alpha$-constant solution. Then obviously $(G,u,v)$ has a solution.

Now suppose that $(G,u,v)$ has a solution $(\{u\}, N(u) \cup S_u, \dots)$. 
Let $t \in S_u$ and let $S_u^* \subseteq S_u$ be a minimum size subset that covers (e.g., dominates) $N(u)$. 
Each $z \in S_u^*$ is connected to $t$ by some path $P_z \subseteq S_u$. Since $G$ is 
$(sP_1+P_4)$-free, $P_z$ has at most $2s+4$ vertices. Hence, $S_v':= \bigcup_{z \in S_u^*} P_z$ has size at most $|S_u^*|(2s+4)$
and is connected and covers $N(u)$ (as $S_u^*$ does). Then $(G,u,v)$ is an $\alpha$-constant solution, unless $|S_u'| > \alpha$.
From now on suppose that $|S_u'| > \alpha$, so in particular $|S_u^*| > s+2$.

We  show that $S_u^*$ is independent.
For contradiction, assume that $z,z'\in S_u^*$ are adjacent. Let $w,w'\in N(u)$ be private neighbours of $z,z'$, resp. Then $wzz'w'$ induces
a $P_4$. Since $G$ is $(sP_1+P_4)$-free, $\{z,z'\}$ must cover almost all vertices in $N(u)$ (which may be assumed independent)
except at most $s-1$ vertices,  say, $w_1, \dots, w_{s-1}$. Thus a minimum size cover $S_u^*$ of $N(u)$ has at most $s+1$ vertices ($z,z'$ and at most
$s-1$ others covering $w_1, \dots, w_{s-1}$), contradicting the fact that $|S_u^*| > s+2$. 

Next we prove that any two vertices $z,z'\in S_u^*$ cover disjoint sets in $N(u)$.
For contradiction, assume that $z,z'\in S_u^*$ have a common neighbour $w \in N(u)$. Since $z\in S_u^*$ also has a private neighbour $w' \in N(u)$,
we find an induced $P_4=w'zwz'$ and conclude that $\{z, z'\}$ must cover all but at most $s-1$ vertices in $N(u)$, a contradiction again.

From the above we conclude that $S_u^* \cup N(u)$ is a disjoint union of stars. 
Recall that $|S_u^*| > s+2> 1$.  Therefore, to be connected, $S_u$ must contain a vertex $t \in S_u \setminus S_u^*$ connecting
two vertices $z,z'\in S_u^*$. Let again $w$ be a private neighbour of $z$ in $N(u)$. Then $wztz'$ is a $P_4$, implying that
$\{t,z,z'\}$ must cover all but $s-1$ vertices in $N(u)$, leading to a contradiction as before. Summarizing, we have shown that 
$(S_u, \dots)$ is an $\alpha$-constant solution.  \qed
\end{proof}

Combining Lemmas~\ref{l-clos} and~\ref{l-constps} gives the desired result:

\begin{theorem}\label{t-sp1p4}
For every constant $s\geq 0$, the {\sc Longest Path Contractibility} problem is polynomial-time solvable for $(sP_1+P_4)$-free graphs.
\end{theorem}

\begin{proof}
By Lemma~\ref{l-constps} we may focus on $\alpha$-constant solutions,
If $(G,u,v)$ has an $\alpha$-constant solution $(\{u\}, N(u) \cup S_u, \dots)$ with $S_u'\subseteq S_u$ of size at most $\alpha$, we may guess
this set $S_u'$ and extend it to its closure $\overline{S'_u}$ (by adding all vertices that are connected to the rest of the graph only 
through $N(u)\cup S_u'$ using Lemma~\ref{l-clos}) in time $O(n^{\alpha+2})$. We may then contract $\{u\} \cup \overline{S'_u}$, thereby reducing
$P_k$-{\sc Suitability} to $P_{k-1}$-{\sc Suitability}.
Since $G$ is $(sP_1+P_4)$-free, we may assume that $k \le 2s+4$. (If $k\geq 2s+5$, then every instance $(G,u,v)$ where
$G$ is $(sP_1+P_4)$-free is a no-instance of $P_k$-{\sc Suitability}). Thus only $2s$ such reductions are required. \qed
\end{proof}

\section{The NP-Complete Cases of Theorem~\ref{t-main}}\label{s-hard}

In this section we prove the new \NP-complete cases of Theorem~\ref{t-main}.

A {\it hypergraph}~${\cal H}$ is a pair $(Q,{\mathcal S})$, where~$Q=\{q_1,\ldots,q_m\}$ is a set of $m$ {\it elements} and~${\mathcal S}=\{S_1,\ldots,S_n\}$ is a set of $n$ {\it hyperedges}, which are subsets of~$Q$.
A {\it $2$-colouring} of ${\cal H}$ is a partition of~$Q$ into two (nonempty) sets $Q_1$ and $Q_2$ with $Q_1\cap S_j \ne\emptyset$ and $Q_2 \cap S_j \ne\emptyset$ for each~$S_j$. This leads to the following decision problem.

\problemdef{Hypergraph 2-Colourability}{a hypergraph ${\cal H}$.}{does ${\cal H}$ have a 2-colouring?}

Note that {\sc Hypergraph 2-Colourability} is \NP-complete even for hypergraphs ${\cal H}$ with $S_i\ne \emptyset$ for $1\leq i\leq n$ and $S_n=Q$. By a reduction from {\sc Hypergraph 2-Colourability}, Brouwer and Veldman~\cite{BV87} proved that $P_4$-{\sc Contractibility} is \NP-complete. That is, from a hypergraph~${\cal H}$ they built a graph $G_{\cal H}$, 
such that ${\cal H}$ has a $2$-colouring if and only if $G_{\cal H}$ has $P_4$ as a contraction.
We first recall the graph $G_{\cal H}$ from~\cite{BV87}, which was obtained from a hypergraph~${\cal H}$ with $S_i\ne \emptyset$ for $1\leq i\leq n$ and $S_n=Q$ (see Figure~\ref{f-hypergraph} for an example).

\begin{itemize}
\item
Construct the {\it incidence graph} of $(Q,{\mathcal S})$, which is 
the bipartite graph with partition classes~$Q$ and~${\cal S}$ and an edge between two vertices~$q_i$ and~$S_j$ if and only if $q_i\in S_j$.
\item
Add a set ${\mathcal S}'=\{S_1',\ldots,S_n'\}$ of~$n$ new vertices, where we call $S_j'$ the {\it copy} of $S_j$.
\item
For $i=1,\ldots m$ and $j=1,\ldots, n$, add an edge between~$q_i$ and~$S_j'$ if and only if $q_i\in S_j$.
\item
For $j=1,\ldots, n$ and $\ell=1,\ldots,n$,
add an edge between $S_j$ and~$S_\ell'$, so the subgraph induced by ${\cal S}\cup {\cal S}'$ will be complete bipartite.
\item 
For $h=1,\ldots, m$ and $i=1,\ldots,m$, add an edge between $q_h$ and $q_i$, so $Q$ will be a clique.
\item Add two new vertices $t_1$ and $t_2$.
\item For $j=1,\ldots,n$, add an edge between $t_1$ and $S_j$, and between $t_2$ and $S_j'$.
\end{itemize}

\begin{figure}
  \centering
  \includegraphics[scale=1]{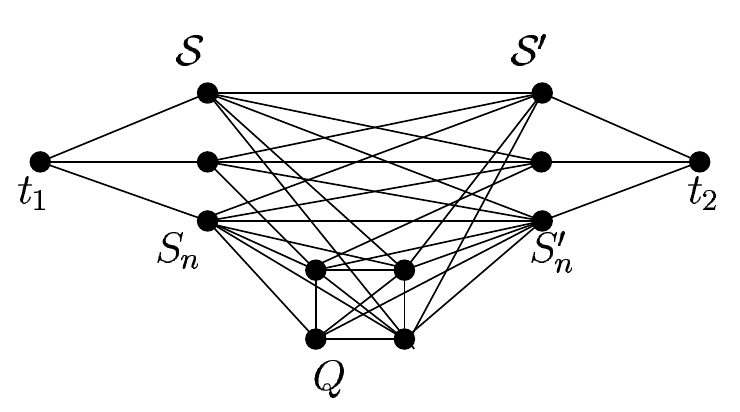}
  \caption{An example of a graph $G_{\cal H}$ for some hypergraph ${\cal H}$~\cite{LPW08}.}\label{f-hypergraph}
\end{figure}

As mentioned, Brouwer and Veldman~\cite{BV87} proved the following.

\begin{lemma}[\cite{BV87}]\label{l-bv87}
A hypergraph ${\cal H}$ has a $2$-colouring if and only if $G_{\cal H}$ has $P_4$ as a contraction.
\end{lemma}

A {\it split graph} is a graph whose vertex set can be partitioned into two (possibly empty) sets
$K$ and $I$, where $K$ is a clique and $I$ is an independent set. 
It is well known that a graph is split if and only if it is $(2P_2,C_4,C_5)$-free~\cite{FH77}.
Let ${\cal H}$ be a hypergraph.
Observe that the subgraphs of $G_{\cal H}$ induced by $Q\cup {\cal S}$ and $Q\cup {\cal S}'$, respectively, are split graphs.
Hence, we make the following observation.

\begin{lemma}\label{l-split}
Let ${\cal H}$ be a hypergraph. Then the subgraphs of $G_{\cal H}$ induced by $Q\cup {\cal S}$ and $Q\cup {\cal S}'$, respectively,
are $(2P_2,C_4,C_5)$-free.
\end{lemma}

We will need the following known lemma from~\cite{HPW09}.

\begin{lemma}[\cite{HPW09}]\label{l-p6}
Let ${\cal H}$ be a hypergraph. Then the graph $G_{\cal H}$ is $P_6$-free.
\end{lemma}

We complement Lemma~\ref{l-p6} with the following lemma.

\begin{lemma}\label{l-freefree}
Let ${\cal H}$ be a hypergraph. Then the graph $G_{\cal H}$ is $(2P_1+2P_2,3P_2, 2P_3)$-free.
\end{lemma}

\begin{proof}
We will prove that $G=G_{\cal H}$ is $(2P_1+2P_2,3P_2, 2P_3)$-free by considering each graph in $\{2P_1+2P_2,3P_2, 2P_3\}$ separately.

\medskip
\noindent
{\bf $\mathbf{(2P_1+2P_2)}$-freeness.}
For contradiction, assume that $G$ contains a subgraph $H$ isomorphic to $2P_1+2P_2$. Let $D_1$ and $D_2$ be the two connected components of $H$ that contain an edge. Let $x$ and $y$ denote the two isolated vertices of $H$.
First suppose that one of $x,y$, say $x$, belongs to ${\cal S}\cup {\cal S}'$, say to ${\cal S}$. 
Then $D_1$ and $D_2$ do not contain~$t_1$ and also do not contain any vertex from ${\cal S}'$. The latter implies that $D_1$ and $D_2$ cannot contain vertex~$t_2$ either. Hence, $D_1$ and $D_2$ only contain vertices from ${\cal S}\cup Q$, contradicting
Lemma~\ref{l-split}. Hence, $x$ and $y$ must both belong to $Q\cup \{t_1,t_2\}$. Suppose one of them, say $x$, belongs to $Q$.
Then $D_1$ and $D_2$ do not contain any vertices from $Q$ and thus only contain vertices from ${\cal S}\cup {\cal S}'\cup \{t_1,t_2\}$.
However, $G[{\cal S}\cup {\cal S}'\cup \{t_1,t_2\}]$ is complete bipartite, and thus $2P_2$-free, a contradiction.
We thus found that $\{x,y\}=\{t_1,t_2\}$. Then $D_1$ and $D_2$ may not contain any vertices from
${\cal S}\cup {\cal S}'$. Consequently, $D_1$ and $D_2$ only contain vertices from $Q$. This is not possible, as $Q$ is a clique.
We conclude that $G$ is $(2P_1+2P_2)$-free.

\medskip
\noindent
{\bf $\mathbf{3P_2}$-freeness.}
For contradiction, assume that $G$ contains a subgraph $H$ isomorphic to $3P_2$. Let $D_1,D_2,D_3$ be the three connected components of $H$. Suppose one of $D_1,D_2,D_3$, say $D_1$, contains a vertex from ${\cal S}\cup {\cal S}'$, say $D_1$ contains a
vertex from ${\cal S}$.
Then $D_2$ and $D_3$ do not contain $t_1$ and also do not contain any vertex from ${\cal S}'$. The latter implies that $D_2$ and $D_3$ cannot contain vertex~$t_2$ either. Hence, $D_2$ and $D_3$ only contain vertices from ${\cal S}\cup Q$, contradicting
Lemma~\ref{l-split}. This means that $H$ contains no vertex from ${\cal S}\cup {\cal S}'$. Consequently, $H$ does not contain $t_1$ and $t_2$ either. However, then $H$ consists of vertices from $Q$ only. This is not possible, as $Q$ is a clique. We conclude that $G$ is $3P_2$-free.

\medskip
\noindent
{\bf $\mathbf{2P_3}$-freeness.}
For contradiction, assume that $G$ contains a subgraph $H$ isomorphic to $2P_3$. Let $D_1$ and $D_2$ be the two connected components of $H$. Suppose one of $D_1,D_2$, say $D_1$, contains a vertex from $Q$. As $Q$ is a clique, this means that
$D_1$ must contain at least one vertex of ${\cal S}\cup {\cal S}'$, say $D_1$ contains a vertex of ${\cal S}$. 
Then $D_2$ cannot contain any vertex from $Q\cup \{t_1\}$ or from ${\cal S}'$. The latter implies that $D_2$ does not contain $t_2$ either. Hence, $D_2$ only contains vertices from ${\cal S}$. This is not possible, as ${\cal S}$ is an independent set.
We conclude that neither $D_1$ nor $D_2$ contains a vertex from $Q$. Hence, $H$ only contains vertices from 
${\cal S}\cup {\cal S}'\cup \{t_1,t_2\}$.
However, $G[{\cal S}\cup {\cal S}'\cup \{t_1, t_2\}]$ is complete bipartite, and thus $2P_3$-free, a contradiction. We conclude that $G$ is $2P_3$-free. \qed
\end{proof}

It is readily seen that $P_4$-{\sc Contractibility} belongs to \NP. Hence, we obtain the following result from 
Lemmas~\ref{l-bv87},~\ref{l-p6}, and~\ref{l-freefree}.

\begin{theorem}\label{t-hard}
{\sc $P_4$-Contractibility} is \NP-complete for $(2P_1+2P_2,3P_2, 2P_3,P_6)$-free graphs.
\end{theorem}

By modifying the graph~$G_{\cal H}$ we prove the next theorem.

\begin{theorem}\label{t-girth}
Let $p\geq 4$ be some constant. Then {\sc $P_{2p}$-Contractibility} is \NP-complete for bipartite graphs of girth at least~$p$.
\end{theorem}

\begin{proof}
We assume without loss of generality that $p$ is even. We reduce again from {\sc Hypergraph $2$-Colouring}, using a suitable subdivision of the graph $G_{\mathcal{H}}$ in order to satisfy the bipartiteness and girth constraints. Let $\mathcal{H}$ be a hypergraph. We first construct  $G_{\mathcal{H}}$. We then subdivide edges in $G_{\mathcal{H}}$ as follows. The edge $S_nS'_n$ is \emph{not}
subdivided. All other edges $S_iS'_j$ are subdivided by an \emph{even} number of vertices, namely by $p-2$, each. 
All edges joining $Q$ to  $\mathcal{S'}$ are also subdivided $p-2$ times. So each of these edges becomes a path of
odd length ($p-1$). Edges joining $Q$ to itself or to $\mathcal{S}$ are subdivided by an \emph{odd} number of vertices,
namely $p-1$, each. So each of these edges becomes a path of even length ($p$).
In addition, we attach paths of length $p-2$, one to each of $t_1$ and $t_2$. Denote these paths by $P_i$ with 
end-vertices
$t_i$ and, say, $\bar{t}_i, i=1,2$. Call the resulting graph $\bar{G}_{\mathcal{H}}$. In what follows we will denote the 
paths of length $p-1$ or $p$ obtained by subdividing an edge $xy$ in $G_{\mathcal{H}}$ by $\overline{xy}$.

The distance between $\bar{t}_1$ and $\bar{t}_2$ in $\bar{G}_{\mathcal{H}}$ equals $2(p-2)+3=2p-1$. The unique shortest path 
is given by $P=(P_1,S_n,S'_n,P_2)$. No other pair of vertices in $\bar{G}_{\mathcal{H}}$ is this far apart. (For example, 
the distance between $S_i \in \mathcal{S}$ and $q \in Q\backslash S_i$ equals $p$ or $2+p$, the length of the path via $t_1$ and
$S_n$.) 

It is straightforward to check that $\bar{G}_{\mathcal{H}}$ is bipartite: Any path joining $\mathcal{S}$ to $\mathcal{S'}$
has odd length. Hence, there are no odd cycles that hit both $\mathcal{S}$ and $\mathcal{S'}$. Similarly, all paths joining 
$Q$ and $\mathcal{S}'$ have odd length. So there cannot be any odd cycle in the subgraph induced by $Q\cup \mathcal{S}'$.
The same argument applies to cycles in the subgraph induced by $Q \cup \mathcal{S}$. Here, again, all paths between $Q$ and 
$\mathcal{S}$ have the same parity (this time even). This shows that $\bar{G}_{\mathcal{H}}$ is indeed bipartite.
The graph $\bar{G}_{\mathcal{H}}$ also has girth at least~$p$. (Recall that any subdivided edge became a path of length at least~$p-1$.) It remains to prove that ${\cal H}$ has a 2-colouring if and only if $\bar{G}_{\mathcal{H}}$ contains $P_{2p}$ as a contraction.

First suppose that ${\cal H}$ has a 2-colouring $(Q_1,Q_2)$.
Then define a contraction of  $\bar{G}_{\mathcal{H}}$
to $P$ with corresponding witness structure $\{W(x), x \in P\}$ as follows. For each $(i,j) \neq (n,n)$ pick any subdivision
vertex $v_{ij} \in \overline{S_iS'_j}$ and let $P_{ij}$ denote the (vertices of the) subpath of $\overline{S_iS_j'}$ from $S_i$ 
to $v_{ij}$.
Similarly,
let $P'_{ij}$ denote the (vertices of) $\overline{S_iS_j'}\backslash P_{ij}$, the "other half" of the path from $S_i$ to $S_j'$.
Now the witnesses can be defined as follows:
\begin{align}\nonumber
 W(t)&= \{t\} ~\text{for} ~t \in P_1 \\ \nonumber
 W(S_n)& = \bigcup_i \{\overline{S_iq}~|~q \in S_i\cap Q_1\} \cup \bigcup_{q,q'\in Q_1}\overline{qq'} \cup 
 \bigcup_{(i,j)\neq(n,n)} P_{ij}\\ \nonumber
 W(S'_n)& = \bigcup_i \{\overline{S'_iq}~|~q \in S'_i\cap Q_2\} \cup \bigcup_{q,q'\in Q_2}\overline{qq'} \cup 
 \bigcup_{(i,j)\neq(n,n)} P'_{ij}\\ \nonumber
 W(t)&= \{t\} ~\text{for} ~t \in P_2 \\ \nonumber
\end{align}
To check correctness, we verify the three conditions for witnesses (observing that disjointness of the bags $W(x)$ is obvious).
\begin{itemize}
 \item $W(S_n)$ is connected: Indeed, all of  $P_{ij}$ is connected to $S_i$ and this (as we assume $Q=Q_1\cup Q_2$ 
 is a $2$-colouring of $\mathcal{H}$) contains some $q \in Q_1$, so $\overline{S_iq}$ joins $S_i$ to $q$. The latter, in turn,
 is joined to $S_n$. The same arguments apply to $W(S'_n)$.
 \item Any two consecutive bags $W(x)$ and $W(y)$ (that is, when $x$ and $y$ are neighbours in $P$) are adjacent:
 Indeed, $t_1$ is adjacent to $S_n$, $S_n$ is adjacent to $S'_n$, and $S'_n$ is adjacent to $t_2$.
 \item If $x$ and $y$ are non-adjacent in $P$, then $W(x)$ and $W(y)$ are non-adjacent in $\bar{G}_{\mathcal{H}}$: Indeed, 
 $t_1$ is only adjacent to $W(S_n)$ and this in turn is only adjacent to $W(S'_n)$  and $t_1$.
\end{itemize}
Thus, indeed, $\bar{G}_{\mathcal{H}}$ contains $P_{2p}$ as a contraction.

Now suppose that $\bar{G}_{\mathcal{H}}$ contains $P_{2p}$ as a contraction. Since $\bar{t}_1$ and $\bar{t}_2$ are the only vertices
at distance $2p-1$ in $\bar{G}_{\mathcal{H}}$, the only possibility is that $\bar{G}_{\mathcal{H}}$ contracts to 
$P=(P_1,S_n, S'_n, P_2)$. Let $\{W(x), x\in P\}$ be a corresponding witness structure. 

\medskip
\noindent
\emph{Claim 1:}\\
(i) $S_i \in W(t_1)\cup W(S_n)$ and $S'_i \in W(t_2)\cup W(S'_n)$ for $i=1, \dots, n$.\\
(ii) $q \in W(t_1) \cup W(S_n)\cup W(S'_n) \cup W(t_2)$ for all $q \in Q$.\\
\emph{Proof of Claim 1.} In order to have all $W(x), x \in P$ connected, the subdivision vertices on $\overline{S_iS'_j}$ must belong
to the same bags $W(x)$ as either $S_i$ or $S'_j$. The vertices $S_i$ and $S'_j$, however, must be in different (adjacent)
bags: Indeed, $S_i, S'_j\in W(x)$ would imply that both $t_1$ and $t_2$ were adjacent to (or contained in) 
$W(x)$, contradicting the third
condition for witness structures. The same argument shows that $S_i$ must either be in $W(t_1)$ or an adjacent bag, 
that is, in $W(S_n)$ or in $W(t)$, where $t$ is the unique neighbour of $t_1$ in $P_1$. The latter, however, is 
impossible: If $S_i \in W(t)$, then $S_i$ must  be connected to $t$ within $W(t)$. But the only path joining $S_i$ to $t$
in $\bar{G}_{\mathcal{H}}$ runs through $t_1$, which does not belong to $W(t)$. Thus, indeed, (i) follows.

Part (ii) can be proved in the same way: If $q\in W(t)$ with $t \in P_1\backslash \{t_1\}$, then $q$ should be connected to 
$t$ within $W(t)$. But, again, the only path connecting $q$ and $t$ runs through $t_1$, a contradiction.\dia

\medskip
\noindent
We claim that the partition $Q=Q_1 \cup Q_2$ given by
$Q_1= Q \cap (W(t_1)\cup W(S_n))$ and $Q_2:= Q \cap (W(t_2) \cup W(S'_n))$ is a $2$-colouring of ${\cal H}$. That is, we will show that each $S_i \in \mathcal{S}$ contains some $q \in Q_1$ and, similarly, each $S'_i \in \mathcal{S'}$ contains some
$q \in Q_2$.
 
Let $S_i \in \mathcal{S}$. From Claim~1 it follows that $S_i \in W(t_1)\cup W(S_n)$. 
For each $q \in S_i$ we follow the path $\overline{S_iq}$ from $S_i$ to $q$ in $\bar{G}_{\mathcal{H}}$. Let $v$ be the last vertex
on this path that belongs to $W(t_1)\cup W(S_n)$. If $v=q$, then $q \in Q \cap (W(t_1)\cup W(S_n)) =Q_1$ and we are done. 
Hence, assume $v\neq q$. Then, in particular, $q \notin W(t_1)\cup W(S_n)$. From Claim 1 we know that 
$q \in W(t_2) \cup W(S'_n)$. The path  $\overline{S_iq}$ starts in $S_i \in W(t_1) \cup W(S_n)$ and ends in 
$q \in W(t_2)\cup W(S'_n)$. Since only $W(S_n)$ and $W(S'_n)$ are adjacent, this path must eventually pass from $W(S_n)$ to
$W(S'_n)$ for the last time. Hence, $v\in W(S_n)$ and, therefore, must be connected to $S_n$ within $W(S_n)$. As $v$ is
a subdivision vertex on $\overline{S_iq}$, this connection can only be via $S_i$ or $q$. But $q$ is not in $W(S_n)$, so the 
connection must be via $S_i$ and we conclude that $S_i \in W(S_n)$. Hence, $S_i$ must be connected to $S_n$ within $W(S_n)$. 
The only paths in $\bar{G}_{\mathcal{H}}$ connecting $S_i$ to $S_n$ run through either $t_1$ (which does \emph{not} belong
to $W(S_n)$) or some  $S'_j$ (which also does not belong to $W(S_n)$) or - the last possibility - some $\tilde{q} \in S_i$.
Hence, indeed, at least one such $\tilde{q} \in S_i$ must belong to $W(S_n)$. But then $\tilde{q} \in Q_1$ (by definition of $Q_1$), as required. \qed
\end{proof}

As a consequence of Theorem~\ref{t-girth}, {\sc Longest Path Contractibility} is \NP-complete for bipartite graphs of arbitrarily large girth. This strengthens the corresponding result for bipartite graphs, which following from a result of~\cite{HHLP14}.
For our dichotomy result we need the following consequence of Theorem~\ref{t-girth}.

\begin{corollary}\label{c-girth}
Let $H$ be a graph that has a cycle. Then {\sc Longest Path Contractibility} is \NP-complete for $H$-free graphs.
\end{corollary}

\begin{proof}
Let $g$ be the girth of $H$. We set $p=g+1$ and note that the class of $H$-free graphs contains the class of graphs of girth at least~$p$. Hence, we can apply Theorem~\ref{t-girth}.\qed
\end{proof}

\section{The Proof of Theorem~\ref{t-main}}\label{s-classification}

We will use the following result from~\cite{FKP13} as a lemma (in fact this result holds even for line graphs which form a subclass of the class of $K_{1,3}$-free graphs).

\begin{lemma}[\cite{FKP13}]\label{l-claw}
The $P_7$-{\sc Contractibility} problem is \NP-complete for $K_{1,3}$-free graphs.
\end{lemma}

By using the results from the previous sections and the above result  we can now prove our classification theorem.

\medskip
\noindent
{\bf Theorem~\ref{t-main}. (restated)}
{\it Let $H$ be a graph. If $H$ is an induced subgraph of $P_1+P_5$, $P_1+P_2+P_3$, $P_2+P_4$ or $sP_1+P_4$ for some $s\geq 0$, then {\sc Longest Path Contractibility} restricted to $H$-free graphs is polynomial-time solvable; otherwise it is \NP-complete.}

\begin{proof}
If $H$ is an induced subgraph of $P_1+P_5$, $P_1+P_2+P_3$, $P_2+P_4$ or $sP_1+P_4$ for some $s\geq 0$, then we use Theorems~\ref{t-p2p4}--\ref{t-sp1p4} to find that {\sc Longest Path Contractibility} is polynomial-time solvable for $H$-free graphs. From now on suppose $H$ is not of this form. Below we will prove that in that case {\sc Longest Contractibility} is \NP-complete for $H$-free graphs.

If $H$ contains a cycle, then we apply Corollary~\ref{c-girth} to prove that {\sc Longest Path Contractibility} is \NP-complete for $H$-free graphs. Assume that $H$ is a forest. If $H$ has a vertex of degree at least~3, then the class of $H$-free graphs contains the class of $K_{1,3}$-free graphs. Hence, we can apply Lemma~\ref{l-claw} to find that {\sc Longest Path Contractibility} is \NP-complete for $H$-free graphs.

From now on we assume that $H$ is a linear forest. As $H$ is not an induced subgraph of $sP_1+P_4$, we find that $H$ contains at least one edge. We distinguish three cases.

\medskip
\noindent
{\bf Case 1.} The number of connected components of $H$ is at least~3.\\
First suppose that at least three connected components of $H$ contain an edge. Then $H$ contains an induced $3P_2$. This means we may apply Theorem~\ref{t-hard} to find that {\sc Longest Path Contractibility} is \NP-complete for $H$-free graphs.

Now suppose that exactly two connected components of $H$ contain an edge. If $H$ contains at least four connected components, then $H$ contains an induced $2P_1+2P_2$. This means we may apply Theorem~\ref{t-hard} to find that {\sc Longest Path Contractibility} is \NP-complete for $H$-free graphs. Hence, $H=P_1+P_r+P_s$ for some $2\leq r\leq s$. If $s\geq 4$, then $H$ contains an induced $2P_1+2P_2$. This means we may apply Theorem~\ref{t-hard} to find that {\sc Longest Path Contractibility} is \NP-complete for $H$-free graphs. If $s=3$ and $r=2$, then $H=P_1+P_2+P_3$, a contradiction. If $s=3$ and $r=3$, then $H$ contains an induced $2P_3$. This means we may apply Theorem~\ref{t-hard} to find that {\sc Longest Path Contractibility} is \NP-complete for $H$-free graphs. Hence, $s=2$, and thus $r=2$. Then $H=P_1+2P_2$ is an induced subgraph of $P_1+P_5$, a contradiction.

Finally suppose that exactly one connected component of $H$ contains an edge. Then $H=sP_1+P_r$ for some $r\geq 2$. As $H$ is not an induced subgraph of $sP_1+P_4$, we find that $r\geq 5$. If $r\geq 6$, then $H$ contains an induced $P_6$. This means we may apply Theorem~\ref{t-hard} to find that {\sc Longest Path Contractibility} is \NP-complete for $H$-free graphs. Hence, $r=5$. As $H\neq P_1+P_5$, we find that $s\geq 2$. Then $H=sP_1+P_5$ contains an induced $2P_1+2P_2$. This means we may apply Theorem~\ref{t-hard} to find that {\sc Longest Path Contractibility} is \NP-complete for $H$-free graphs.

\medskip
\noindent
{\bf Case 2.} The number of connected components of $H$ is exactly~2.\\ 
Then $H=P_r+P_s$ for some $r$ and $s$ with $1\leq r\leq s$.  If $r\geq 3$ then $H$ contains an induced $2P_3$, and if $s\geq 6$ then $H$ contains an induced $P_6$. In both cases we apply Theorem~\ref{t-hard} to find that {\sc Longest Path Contractibility} is \NP-complete for $H$-free graphs. From now on assume that $r\leq 2$ and $s\leq 5$. If $s\leq 4$, then $H$ is an induced subgraph of $P_2+P_4$, a contradiction. Hence, $s=5$. If $r=1$, then $H=P_1+P_5$, a contradiction. Thus $r=2$. Then $H$ contains an induced $3P_2$. This means we may apply Theorem~\ref{t-hard} to find that {\sc Longest Path Contractibility} is \NP-complete for $H$-free graphs.

\medskip
\noindent
{\bf Case 3.} The number of connected components of $H$ is exactly~1.\\
If $H=P_r$ for some $r\leq 5$, then we use Theorem~\ref{t-p1p5}. Otherwise $P_6$ is an induced subgraph of~$H$, and we use Theorem~\ref{t-hard}. \qed
\end{proof}

\section{Longest Cycle Contractibility}\label{s-cycle}

The length of of a longest cycle a graph~$G$ can be contracted to is called the {\it co-circularity}~\cite{Bl82} or 
{\it cyclicity}~\cite{Ha99} of~$G$. This leads to the following decision problem.

\problemdef{Longest Cycle Contractibility}{a connected graph $G$ and an integer~$k$.}{does $G$ contain $C_k$ as a contraction?}
Hammack proved that  {\sc Longest Cycle Contractibility} is \NP-complete for general graphs~\cite{Ha02} but polynomial-time solvable for planar graphs~\cite{Ha99}. It is also known that $C_6$-{\sc Contractibility}, and thus {\sc Longest Cycle Contractibility}, is \NP-complete for $K_{1,3}$-free graphs~\cite{FKP13} and bipartite graphs~\cite{DP17}, and thus for $C_r$-free graphs if $r$ is odd.
The purpose of this section is to show that the complexities of {\sc Longest Cycle Contractibility} and {\sc Longest Path Contractibility} may not coincide on $H$-free graphs. 

For a given hypergraph~${\cal H}=(Q,{\cal S})$ we first construct the graph $G_{\cal H}$ as before. We then add an edge between vertices~$t_1$ and $t_2$. This yields the graph~$G_{\cal H}'$.
We need the following result from~\cite{BV87}.

\begin{lemma}[\cite{BV87}]\label{l-bv87b}
A hypergraph ${\cal H}$ has a $2$-colouring if and only if $G_{\cal H}'$ has $C_4$ as a contraction.
\end{lemma}

We now prove the following lemma.

\begin{lemma}\label{l-p2p4again}
Let ${\cal H}$ be a hypergraph. Then the graph $G_{\cal H}'$ is $(P_2+P_4)$-free.
\end{lemma}

\begin{proof}
For contradiction, assume that $G_{\cal H}'$ contains a subgraph $H$ isomorphic to $P_2+P_4$. Let $D_1$ be the connected component of $H$ on two vertices, and let $D_2$ be the connected components of $H$ on four vertices.
As the subgraph of $G_{\cal H}'$ induced by ${\cal S}\cup {\cal S}'\cup \{t_1,t_2\}$ is complete bipartite and thus $P_4$-free, we find that $D_2$ must contains a vertex of $Q$. As $Q$ is a clique, $D_2$ must also contain at least two vertices from ${\cal S}\cup {\cal S}'$. We may without loss of generality assume that $D_2$ contains a vertex from ${\cal S}$. 
This means that $D_1$ contains no vertex from $Q\cup {\cal S}'\cup \{t_1\}$. Hence, $D_1$ only contains vertices of ${\cal S}\cup \{t_2\}$. As the latter set is independent, this is not possible.
We conclude that $G_{\cal H}'$ is $(P_2+P_4)$-free. \qed
\end{proof}

We note that, in line with our polynomial-time result of {\sc Longest Path Contractibility} for $(P_2+P_4)$-free graphs (Theorem~\ref{t-p2p4}),  the graph $G_{\cal H}$ may not be $(P_2+P_4)$-free: as $t_1$ and $t_2$ are not adjacent in $G_{\cal H}$, two vertices of $Q$ together with vertices $t_1,S_j,S_\ell',t_2$ may form an induced $P_2+P_4$ in $G_{\cal H}$.

It is readily seen that $C_4$-{\sc Contractibility} belongs to \NP. Hence, we obtain the following result from Lemmas~\ref{l-bv87b} and~\ref{l-p2p4again}.

\begin{theorem}\label{t-cycle}
{\sc $C_4$-Contractibility} is \NP-complete for $(P_2+P_4)$-free graphs.
\end{theorem}

Theorem~\ref{t-cycle} has the following consequence.

\begin{corollary}\label{c-cycle}
{\sc Longest Cycle Contractibility} is \NP-complete for $(P_2+P_4)$-free graphs.
\end{corollary}

Recall that {\sc Longest Path Contractibility} is polynomial-time solvable for $(P_2+P_4)$-free graphs by Theorem~\ref{t-main}. Hence, combining this result with Corollary~\ref{c-cycle} shows that the two problems behave differently on $(P_2+P_4)$-free graphs.

\section{Conclusions}\label{s-con}

We completely classified the complexities of {\sc Longest Induced Path} and {\sc Longest Path Contractibility} problem for $H$-free graphs. Such a classification is still open for {\sc Longest Path} and below we briefly present the state of art.

A graph is {\it chordal bipartite} if it is bipartite and every induced cycle has length~4. In other words, a graph is chordal bipartite if and only if it is $(C_3,C_5,C_6,\ldots)$-free.
A direct consequence of the \NP-hardness of {\sc Hamiltonian Path} for chordal bipartite graphs and strongly chordal split graphs~\cite{Mu96}, or equivalently, strongly chordal $(2P_2,C_4,C_5)$-free graphs~\cite{FH77}
is that 
{\sc Hamiltonian Path}, and
therefore, {\sc Longest Path} is \NP-complete for $H$-free graphs if $H$ has a cycle or contains an induced $2P_2$. The \NP-hardness of {\sc Hamiltonian Path} for line graphs~\cite{Be81}, and thus for $K_{1,3}$-free graphs, implies the same result for $H$-free graphs if $H$ is a forest with a vertex of degree at least~3. On the positive side, {\sc Longest Path} is polynomial-time solvable for $P_4$-free graphs due to the corresponding result for its superclass of cocomparability graphs~\cite{IN13,MC12}.
This leaves open the following cases.
\begin{oproblem}
Determine the computational complexity of {\sc Longest Path} for $H$-free graphs when:
\begin{itemize}
\item $H=sP_1+P_r$ for $3\leq r\leq 4$ and $s\geq 1$
\item  $H=sP_1+P_2$  for $s\geq 2$
\item  $H=sP_1$  for
$s\geq 3$.
\end{itemize}
\end{oproblem}

We showed that the complexities of {\sc Longest Cycle Contractibility} and {\sc Longest Path Contractibility} do not coincide for $H$-free graphs. However, the complexity of {\sc Longest Cycle Contractibility} for $H$-free graphs has not been settled yet. 
For instance, if $H$ is a cycle, the cases $H=C_4$ and $H=C_6$ are still open.

\bibliographystyle{abbrv}

\end{document}